\documentclass[american,aps,prx,reprint,floatfix,nofootinbib,superscriptaddress,notitlepage]{revtex4-2}
\usepackage[unicode=true,pdfusetitle, bookmarks=true,bookmarksnumbered=false,bookmarksopen=false, breaklinks=false,pdfborder={0 0 0},backref=false,colorlinks=false]{hyperref}
\hypersetup{colorlinks,linkcolor=myurlcolor,citecolor=myurlcolor,urlcolor=myurlcolor}
\usepackage{graphics,epstopdf,graphicx,amsthm,amsmath,amssymb,
braket,colortbl,color,bm,framed,mathrsfs}
\usepackage[T1]{fontenc}
\usepackage[up]{subfigure}
\usepackage{tikz}
\usepackage{algorithm}
\usepackage{algorithmic}
\usepackage{enumerate}
\usepackage{tcolorbox}
\usepackage[toc,page,header]{appendix}

\definecolor{myurlcolor}{rgb}{0,0,0.9}

\newcommand{\proj}[1]{\left| #1\right\rangle\!\left\langle #1 \right|}

\DeclareMathOperator{\trace}{Tr}
\newcommand{\Ptr}[2]{\trace_{#1}\Pa{#2}}
\newcommand{\Tr}[1]{\Ptr{}{#1}}

\newcommand{\Pa}[1]{\left[#1\right]}

\newcommand{\norm}[1]{\left\lVert #1 \right\rVert}

\theoremstyle{plain}
\newtheorem{thm}{Theorem}

\newtheorem{prop}[thm]{Proposition}

\newtheorem{Def}[thm]{Definition}

\newcommand*{\myproofname}{Proof}

\def\ot{\otimes}
\def\complex{\mathbb{C}}
\def\real{\mathbb{R}}

\DeclareMathAlphabet{\mathcal}{OMS}{cmsy}{m}{n}

\let\oldaddcontentsline\addcontentsline
\newcommand{\stoptocentries}{\renewcommand{\addcontentsline}[3]{}}
\newcommand{\starttocentries}{\let\addcontentsline\oldaddcontentsline}
\stoptocentries

\newtheorem{theorem}{Theorem}
\newtheorem{definition}{Definition}

\newtheorem{lemma}[theorem]{Lemma}
\newtheorem{proposition}[theorem]{Proposition}

\newtheorem{corollary}{Corollary}

\numberwithin{equation}{section}
\numberwithin{corollary}{section}
\numberwithin{definition}{section}
\numberwithin{theorem}{section}
\numberwithin{remark}{section}
\numberwithin{example}{section}

\newcommand{\ra}{\rangle}
\newcommand{\la}{\langle}
\newcommand{\df}{\dfrac}

\newcommand{\R}{\mathbb R}

\newcommand{\mbb}{\mathbb}

\newcommand{\mbf}{\mathbf}

\newcommand{\mrm}{\mathrm}

\newcommand{\tr}{\mathrm{Tr}}
\newcommand{\mca}{\mathcal}
\newcommand{\Hi}{\mca H}
\newcommand{\G}{\mca G}

\newcommand{\Pf}{\mathrm{Pf}}
\newcommand{\Cl}{\mca C}

\newcommand{\leqalign}[2]{
\begin{equation}\begin{aligned}\label{#1}
#2
\end{aligned}\end{equation}}
\newcommand{\malign}[1]{
\[\begin{aligned}
#1
\end{aligned}\] 
}

\makeatother
\begin{document}

  \author{Xingjian Lyu}
  \email{nicholaslyu@college.harvard.edu}
\affiliation{ Harvard University, Cambridge, Massachusetts 02138, USA}

  \author{Kaifeng Bu}
  \email{bu.115@osu.edu}
\affiliation{Ohio State University, Columbus, Ohio 43210, USA}
 \affiliation{ Harvard University, Cambridge, Massachusetts 02138, USA}

\title{Fermionic Gaussian Testing and Non-Gaussian Measures via Convolution}

\begin{abstract} 
    We define fermionic convolution and demonstrate its utility in characterizing 
    fermionic non-Gaussian components, which are essential to the 
    computational advantage of fermionic systems. 
    Using fermionic convolution, 
    we propose an efficient protocol that tests the fermionic Gaussianity 
    of pure states using three copies of the input state. 
    We also introduce ``Non-Gaussian Entropy,'' an experimentally 
    accessible resource measure that quantifies fermionic 
    non-Gaussianity. These results provide new insights into the 
    study of fermionic quantum computation.
\end{abstract}

\maketitle
\section{Introduction}
\stoptocentries

A central challenge in quantum science is understanding 
the boundary between quantum and classical computation. 
One strategy is to explore the classical simulability 
of quantum computation models and identifying the resources that 
enable computational advantages. Characterizing these resources is essential 
for evaluating the potential of quantum information-processing systems.

For instance, the Gottesman-Knill theorem~\cite{gottesman1998heisenberg} 
reveals that stabilizer circuits—composed of stabilizer states, 
Clifford gates, and Pauli measurements—can be efficiently simulated classically. 
This underscores the critical role of nonstabilizer components, or ``magic,''
in achieving universal quantum computation~\cite{Bravyi06}. 
The amount of magic can be shown to bound the classical simulation 
time of quantum circuits~\cite{BravyiPRL16,BravyiPRX16,bravyi2019simulation,BeverlandQST20,SeddonPRXQ21, bu2022classical,Bu19,koh2015further}.

In continuous-variable bosonic systems, 
Gaussian states and unitaries play a similar foundational role in quantum 
information protocols due to their mathematical tractability~\cite{SethRMP12}. 
It is well established that bosonic Gaussian circuits 
consisting of Gaussian states, unitaries, and measurements 
can be efficiently classically simulated~\cite{Mari12}, 
implying that non-Gaussianity is a key quantum resource for 
computational advantage. This is also evidenced by the utility 
of non-Gaussian components in numerous 
quantum information tasks~\cite{EisertPRL02,Giedke02,GarcPRL04,CerfPRL04,BaragiolaPRL19,JiangPRL24}, 
such as error correction in the GKP code~\cite{GKP01}. 
Consequently, the resource theory of continuous-variable 
bosonic non-Gaussianity has been developed to characterize 
its role in bosonic quantum 
computation~\cite{GenoniPRA08,GenoniPRA10,Marian13,Albarelli18,ZhuangPRA18,TakagiPRA18,LamiPRA18,Chabaud20,WalschaersPRXQ21,Dias24,Han24,Oliver24b}.

Another important framework for the classical simulation of quantum circuits is matchgate simulation, first introduced by Valiant~\cite{valiant2002quantum}. It has been shown that quantum circuits composed of matchgates acting on nearest-neighbor qubits can be  simulated efficiently on a classical computer~\cite{valiant2002quantum,terhal2002classical,jozsa2008matchgates,brod2016efficient}.  
The connection between matchgate circuits and the theory of fermionic linear optics (equivalently, fermionic Gaussian circuits) has been later explored in~\cite{terhal2002classical,divincenzo2004fermionic,jozsa2008matchgates,bravyi2004lagrangian}. The theory of fermionic linear optics 
has a number of applications, including Majorana
fermions in quantum wires~\cite{Kitaev_2001} and Kitaev’s honeycomb
lattice model~\cite{Kitaev_2006}. 
Supported by advancements in experimental platforms, 
this theory continues to drive advances in quantum information 
processing~\cite{GiorginiRMP08,jordens2008mott,LossPRA98,HansonRMP07}. 

Recent work has extended the notion of magic to fermionic computation; 
it has been shown that all pure fermionic non-Gaussian states exhibit magic and, 
with adaptive measurement, 
enable universal quantum computation~\cite{hebenstreit2019all}; 
this result parallels that for discrete bosonic Gaussian, or Clifford, circuits. 
General resource measures have been applied to magic-enhanced 
fermionic Gaussian circuits~\cite{dias2024classical,Oliver24,Joshuar24}, 
contributing to an understanding of their classical simulation.

To gain more insights into fermionic non-Gaussianity, which is central 
to the computational advantage of fermionic quantum computation, we propose a systematic approach using fermionic convolution. This is motivated by the idea that
Gaussian properties are uniquely invariant under convolution.
We define and study the notion of fermionic convolution, which is defined using a fermionic Gaussian unitary. The iterative application of fermionic convolution to a fermionic input state converges to the Gaussian state with the same
covariance; we call the process of this iterative-limit fermionic
Gaussification. We prove that fermionic Gaussification
maps a fermionic state to  the closest Gaussian state in terms of relative entropy, and that this relative entropy decreases monotonically with each iterative application of convolution. We also establish fermionic Gaussification as a resource-destroying map in the resource theory
of fermionic quantum computation.


We prove that fermionic self-convolution preserves entropy only for Gaussian states. Based on this property, we propose a new measure, “Non-Gaussian
Entropy.''to quantify non-Gaussianity. The problem of Gaussian property testing is also considered. Property testing is a crucial area in quantum computation that entails assessing whether a given 
state or circuit has a certain property. For instance, separability testing is used to determine whether a state is
separable ~\cite{Harrow13,v011a003,Beckey21,gs007, Buhrman01}. Here, we propose a novel protocol based on fermionic convolution to distinguish fermionic Gaussian states which only uses three copies of the input state. This protocol readily extends to Gaussian unitary testing and provide valuable tools for 
identifying quantum computational resources.


\subsection{Preliminaries}

Let us consider a system of $n$ fermionic modes with creation $a^\dag_j$
and annihilation $a_j$ operators for $j=1,...,n$. They satisfy the canonical anti-commutation relations 
\begin{eqnarray}
    \{ a_j, a_k^\dag \} = \delta_{jk}, \quad \{a_j, a_k\} = 0. 
\end{eqnarray}
It is convenient to define $2n$ Hermitian Majorana operators $\set{\gamma_j}^{2n}_{j=1}$ defined as 
\begin{eqnarray}
    \gamma_{2j-1}=a_j+a^\dag_j, \quad\gamma_{2j}=i(a_j-a^\dag_j).
\end{eqnarray}
The operators $\set{\gamma_j}^{2n}_{i=j}$ are Hermitian, traceless and generate a Clifford algebra $\Cl_{2n}$ which satisfies the relation 
\leqalign{eq:cliffordRelation}{
    \{\gamma_j,\gamma_k\}=2\delta_{jk}I.
}
The Majorana operators $\set{\gamma_j}^{2n}_{j=1}$  can be represented by the product of Pauli operators 
by the Jordan-Wigner transformation as follows
\begin{align*}
    \gamma_{2j-1}&= Z^{\otimes (j-1)} \otimes X \otimes I^{\otimes (n-j)}, \\ 
    \gamma_{2j} &= Z^{\otimes(j-1)}\otimes Y \otimes I^{\otimes (n-j)}, ~~\forall j.
\end{align*}

Given an ordered subset (multi-index) $J\subset [2n]$, 
we define $\gamma_J$ as the ordered product of the Majorana operators indexed by $J$. Specifically, 
\begin{align*}
    \gamma_J=\prod_{j\in J}\gamma_j.
\end{align*}
We use $|J|$ to denote the size of the subset $J$.
Any linear operator $A\in \Cl_{2n}$ 
can be expanded in the Majorana basis: 
\begin{eqnarray}
    \label{eq:majoranaExpansion}
    A = \df 1 {2^n} \sum_J A_J \gamma_J, \quad A_J=\tr(\gamma_J^\dag A)\in \mbb C. 
\end{eqnarray}
An element $A$ is even if all coefficients $A_J$=0 for $|J|$ odd. 
For example, under the Jordan-Wigner transformation, 
$|1\ra \la 1|=(1-Z)/2$ is even since $Z=-i\gamma_{1}\gamma_{2}$, but 
$|0\ra \la 1|=(X+iY)/2=(\gamma_1 + i\gamma_2)/2$ is not. 
Since an even element $A$ commutes with the fermionic number-parity operator $\Pi^{2n}_{j=1}\gamma_j$, the parity of the number of fermions is conserved by the action of an even operator. 

The covariance matrix $\Sigma_\rho$, or $\Sigma(\rho)$, of a state $\rho\in C_{2n}$ is the $2n\times 2n$ real, anti-symmetric matrix with entries 
\leqalign{eq:covariance}{
    (\Sigma_\rho)_{jk}
    = \frac{i}{2}\Tr{\rho[\gamma_j, \gamma_k]}.
}
It can be brought into the block-diagonal form by a real
rotation $R\in SO(2n)$, i.e, 
\begin{equation}\label{eq:block_D}
    \Sigma_{\rho}=R\bigoplus^n_{j=1}
    \left[
    \begin{array}{cc}
        0 & \lambda_j \\
       -\lambda_j  & 0
    \end{array}
    \right] R^T.
\end{equation}
Fermionic Gaussian states $\rho$ are defined as even states of 
the following form 
\leqalign{def:Gaussian}{
    \rho=C\exp\left(\frac{i}{2} \gamma^T h \gamma\right),
}
where $\gamma$ abbreviates the vector $(\gamma_1, \gamma_2,...,\gamma_{2n})$, $(h_{jk})$ is a real antisymmetric matrix, and $C\in \R$ is a normalization factor; a pure Gaussian state is defined as the limit of a converging sequence of the form above. 
Fermionic Gaussian states represent the ground and thermal states of non-interacting fermionic systems.
By block-diagonalizing $h$ using a rotation $R$, we can rewrite 
\begin{align*}
    \rho=\frac{1}{2^n}\prod^n_{j=1}(1+i\lambda_j \tilde{\gamma}_{2j-1}\tilde{\gamma}_{2j}),
\end{align*}
where $\tilde{\gamma}=R\gamma$, and $\lambda_j$ are the eigenvalues in the block diagonalization~\eqref{eq:block_D}.
Each $\lambda_j$ lies in the interval $[-1,1]$. For Gaussian pure states, each $\lambda_j\in \{-1,1\}$.

In the theory 
of fermionic quantum computation,  a special role is played by those unitary transformations which are generated by Hamiltonians of non-interacting fermionic systems, i.e. Hamiltonians which are quadratic in Majorana fermion operators. Specifically, 
 \begin{eqnarray}
        \label{def:GaussianUnitary}
        U = \exp \left(\df 1 2 \gamma^T h \gamma\right),
    \end{eqnarray}
where $h$ is an anti-symmetric matrix. 
These unitaries are known 
as fermionic Gaussian unitaries (or fermionic linear optics transformations). Conjugation by a Gaussian unitary effect an automorphism of the algebra $\Cl_{2n}$ by rotating the Majorana operators (\cite{jozsa2008matchgates}, Theorem 3), 
\leqalign{eq:unitaryEffect}{
 U \gamma_j U^\dag = \sum_k R_{jk} \gamma_k, \quad R=\exp(2h) \in SO(2n). 
}
See Appendix A for details on the Clifford algebra, and Appendix B for Gaussian states and unitaries.

\section{Fermionic convolution}
Let us begin by introducing the concept of a fermionic beam splitter, which mixes  the Majorana operators of fermionic modes. This concept is inspired by the analogous device in bosonic systems, the bosonic beam splitter.

A fermionic beam splitter, or convolution unitary, acting on $2n$ modes is represented by the Gaussian unitary 
\leqalign{eq:beamSplitterUnitary}{
W_{\theta}= \exp \left(\df \theta 2 \sum_{j=1}^{2n} \gamma_j \gamma_{2n+j}\right). 
}
The unitary $W_\theta$ rotates the Majorana operators by 
\leqalign{eq:convolutionUnitaryAction}{
    W_\theta\gamma_j W_\theta^\dag=\cos\theta \gamma_j-\sin\theta \gamma_{2n+j},\\
    W_\theta\gamma_{2n+j} W_\theta^\dag=\sin\theta \gamma_j+\cos\theta \gamma_{2n+j},
}
for every $1\leq j\leq 2n$ and $\theta\in [-\pi, \pi]$. 
The $\theta=\pi/4$ case corresponds to the balanced fermionic beam splitter. 

\begin{Def}[\bf Convolution]
The fermionic convolution of two even states $\rho, \sigma$ with angle $\theta\in (-\pi, \pi]$ is
\leqalign{def:evenConvolution}{
    \rho \boxtimes_\theta \sigma = \Ptr{2}{W_\theta(\rho \otimes \sigma) W_\theta^\dag},
}
where $\Ptr{2}{\cdot}$ denotes partial trace over the second register corresponding to the $n$ qubits of $\sigma$. 
\end{Def}

Fermionic convolution satisfies several desirable properties: 
if both $\rho$ and $\sigma$ are Gaussian, then $\rho\boxtimes\sigma$ is Gaussian; 
convolution also commutes with any Gaussian unitary $U_G$, i.e., 
$(U_G\rho U^\dag_G)\boxtimes_{\theta} (U_G\sigma U^\dag_G)=U_G(\rho\boxtimes_{\theta}\sigma)U^\dag_G$. 
Additional properties and their detailed proofs can be found in Appendix C. 
For simplicity, we focus $\theta=\pi/4$ throughout the main discussion, abbreviating $\boxtimes_{\pi/4}$ as $\boxtimes$.

Let us consider iterative self-convolution $\boxtimes^k\rho$ for any integer $k$, 
defined according to $\boxtimes^{k}\rho = (\boxtimes^{k-1} \rho)\boxtimes (\boxtimes^{k-1}\rho)$ and $ \boxtimes^0\rho = \rho$.
Based on a known central limit theorem~\cite{Hudson_1973,hudson1980translation}, the operation $\boxtimes^k\rho$ will converge to the Gaussian state with the same covariance. 
We define this iterative limit the \textit{fermionic Gaussification} $\mca G(\rho)$ of $\rho$; it maps a fermionic state $\rho$ to a Gaussian state with the same covariance. We provide two proofs of the central limit theorem with a quantitative bound on the rate of convergence in Appendix D.

We prove that the fermionic Gaussification $\mca G(\rho)$ is in fact the Gaussian projection of $\rho$: 
it is the closest fermionic Gaussian state to the given even state $\rho$ in terms of the quantum relative entropy 
$D(\rho||\sigma)=\tr(\rho\log\rho)-\tr(\rho\log \sigma)$.

\begin{prop}\label{prop:ent_cha}
For any even state $\rho$, we have
\begin{align}
\min_{\sigma_G\in \textbf{Gauss} }D(\rho||\sigma_G)
=D(\rho||\mca G(\rho))
=S(\mca G(\rho))-S(\rho),
\end{align}
$\textbf{Gauss}$ denotes the set of all fermionic Gaussian states. 
\end{prop}

\begin{figure}[t]
  \center{\includegraphics[width=6cm]  {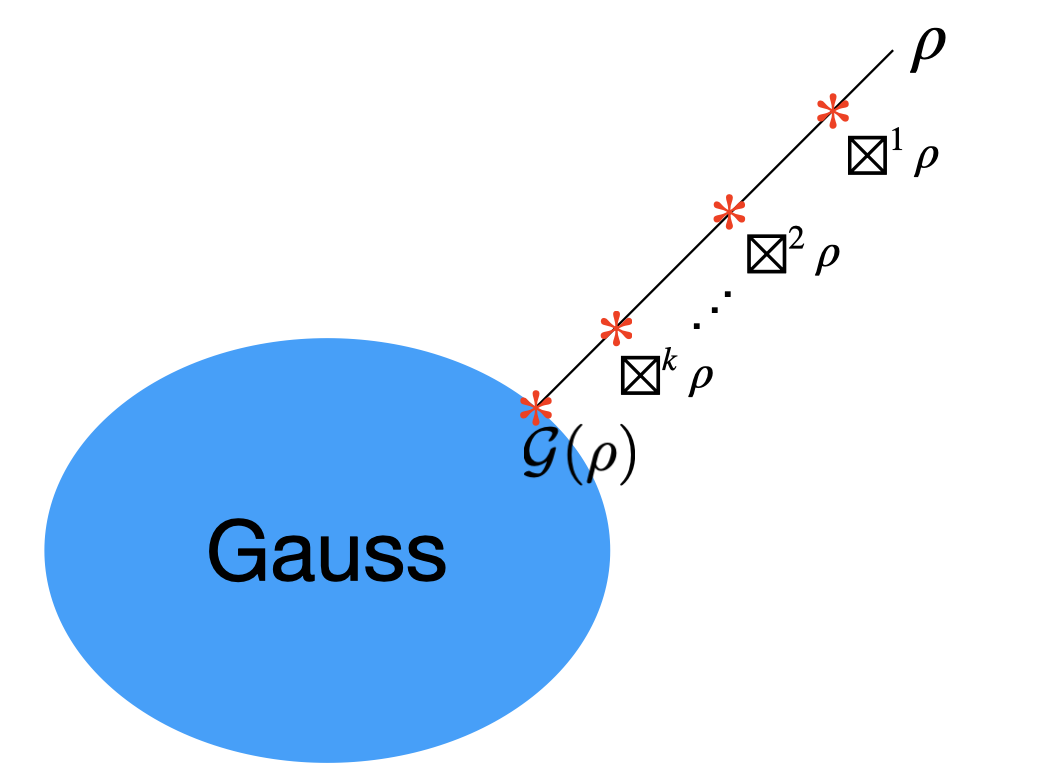}}     
  \caption{The Gaussification $\mca G(\rho)$ is the closest state to the states $\rho$ and $\boxtimes^k\rho$.}
  \label{fig:Clos}
 \end{figure}

If we take the set $\textbf{Gauss}$ to be the set of 
free states $\mathcal{F}$ in the resource theory of fermionic quantum computation, 
we can define the relative entropy of fermionic non-Gaussianity as 
\begin{eqnarray}
    NG_R(\rho):=\min_{\sigma_G\in \textbf{Gauss} }D(\rho||\sigma_G),
\end{eqnarray}
which also equals $S(\mca G(\rho))-S(\rho)$ by Proposition~\ref{prop:ent_cha}.
It is easy to show that $NG_R$ satisfies the following desirable properties which 
establishes it as a magic measure in the resource theory of fermionic quantum computation. 
\begin{prop}\label{prop:NGR}
The following properties hold for the relative entropy of fermionic non-Gaussianity $NG_R$ , 
    \begin{enumerate}[(1)]
    \item Faithfulness: $NG_R(\rho)\geq 0$ with equality iff $\rho$ is a fermionic Gaussian. 
    \item Gaussian-invariance: $NG_R(U_G\rho U^\dag_G)=NG_R(\rho)$ for any Gaussian unitary $U_G$. 
    \item Additivity under tensor product: $NG_R(\rho_1\ot\rho_2)=NG_R(\rho_1)+NG_R(\rho_2)$. 
\end{enumerate}
\end{prop}

Recall the foundational concept of a resource destroying map in resource theory~\cite{LiuZiwenPRL17}: 
A map $\lambda$ from states to states is resource-destroying if 
it satisfies the following two conditions: i) it maps all states to free states, i.e., $\lambda(\rho)\in \mca F$ for any quantum state $\rho$, and 
ii) it fixes the free states, i.e., $\lambda(\sigma)=\sigma$ for any $\sigma\in \mca F$.
Natural resource-destroying maps are known in resource theories of coherence~\cite{BaumgratzPRL14,Plenio17}, asymmetry~\cite{Gour09}, bosonic non-Gaussianity~\cite{Marian13}, and magic in the stabilizer quantum computation~\cite{BGJ23a}. 
It is easy to check that 
the Gaussification $\mca G(\cdot)$ satisfies the conditions i) and ii) with $\mathcal{F}=\textbf{Gauss}$, which motivates the following statement: 
\begin{prop} Gaussification $\mca G(\cdot)$  is a (nonlinear) resource destroying map
in the resource theory of fermionic quantum computation with free states $\textbf{Gauss}$. 
\end{prop}

Moreover, $\mca G(\rho)$ is also the Gaussification of self-convolution state $\boxtimes^k\rho$, i.e.,$\mca G(\boxtimes^k\rho)=\mca G(\rho)$, the divergence can be shown to be monotonically decreasing: $D(\boxtimes^{k}\rho||\mca G(\rho))\geq D(\boxtimes^{k+1}\rho||\mca G(\rho))$ for any integer 
$k$ (See Figure.~\ref{fig:Clos}). This is the result of the following theorem, which we prove  in Appendix C. 
\begin{thm}
For any even state $\rho$, we have 
   \begin{align}
    S(\boxtimes^{k+1}\rho)\geq S(\boxtimes^k\rho), \quad \forall k.
\end{align} 
\end{thm}

\section{Fermionic Gaussian testing}
We have studied the basic properties of fermionic convolution. We now explore the applications of this tool to 
fermionic Gaussian testing. The question is: given an pure fermionic state $\psi$, how to decide whether $\psi$ is a Gaussian state? This is an important question because all non-Gaussian fermionic states have been shown to be magic states, which can be used to generate universal quantum computation~\cite{hebenstreit2019all}.

We introduce a protocol to test fermionic Gaussianity utilizing three copies of the state. This protocol consists of fermionic convolution and a swap test to measure the overlap between the self-convolved state $\psi \boxtimes \psi$ and $\psi$. 

\begin{center}
  \begin{tcolorbox}[width=8cm,height=3.6cm,title=Fermionic Gaussian Test]
\begin{enumerate}
    \item Prepare 2 copies of the pure states $\psi$ and compute the self-convolution $\psi\boxtimes\psi$.

    \item Perform the swap test between $\psi\boxtimes\psi$ and another copy $\psi$.  If the output is $0$, it passes the test; otherwise, it fails.
\end{enumerate}
\end{tcolorbox}
\end{center}
The probability of acceptance in the above fermionic Gaussian test is equal to 
\begin{eqnarray}
\text{Pr}_{\text{accept}}(\psi)
=\frac{1}{2}\left[1+\bra{\psi}\psi\boxtimes\psi\ket{\psi}
\right].
\end{eqnarray}

\begin{thm}
For any pure even state $\ket{\psi}$, the pure state is a Gaussian state if and only if the 
the probability of acceptance  $\text{Pr}_{\text{accep}}(\psi)=1$.
\end{thm}
This theorem results from the fact that self-convolution is entropy-preserving if and only if the input is Gaussian. In particular, for a pure even state $\psi$, $\psi \boxtimes \psi$ remains pure if and only if $\psi$ is a fermionic Gaussian state; 
the detailed proof is provided in Appendix C. We also extend the Gaussian test from states to unitaries using the Choi–Jamiołkowski isomorphism in Appendix G. 

It is worth recalling that separability testing relies on the principle that a pure bipartite state $\ket{\phi}_{AB}$ is separable if and only if the reduced state $\rho_A$ is pure. This principle underpins the well-known entanglement measure: entanglement entropy $S(\rho_A)$. The fermionic Gaussian test motivates a new quantitative measure for fermionic non-Gaussianity, which we call ``non-Gaussian entropy.''

\textbf{\textit{Non-Gaussian entropy:}} 
The non-Gaussian entropy of any pure even state $\psi$ is defined as
\leqalign{eq:ngEntropyBase}{
    NG(\psi)=S(\psi\boxtimes\psi).
}

\begin{thm}
\label{thm:NGProperties}
 The non-Gaussian entropy for pure even states $\psi$ also satisfy properties (1)-(3) in~ Proposition~\ref{prop:NGR}. 
\end{thm}
The proof is provided in Appendix C. This justifies the use of non-Gaussian entropy as a resource measure in the theory of fermionic quantum computation.

We can also generalize the non-Gaussian entropy
to higher orders. 
The $k$-th order non-Gaussian entropy
\leqalign{eq:ngEntropyInductive}{
    NG^{(k)}(\psi)=S(\boxtimes^k\psi)
}
also satisfies the above three properties in Theorem~\ref{thm:NGProperties}. 
Moreover, the asymptotic behavior of higher-order non-Gaussian entropy approaches the relative entropy of non-Gaussianity $NG_R$, i.e,
\begin{eqnarray}
    \lim_{k\to \infty}NG^{(k)}(\psi)
    =NG_R(\psi).
\end{eqnarray}
This demonstrates that non-Gaussian entropy is a finite-shot approximation 
of Gaussian projection divergence.

\textit{Example:} Let us consider a family of quantum states on a 4-qubit system,
\leqalign{def:4QubitMagic}{
    \ket{\psi_{\phi}}=\frac{1}{2}
    \left(\ket{0000}+\ket{0011}+\ket{1100}+e^{i\phi}\ket{1111} \right),
}
with $\phi\in (0,2\pi)$.  Theses state  are shown to be magic states for fermionic quantum computation~\cite{hebenstreit2019all}. We compute the non-Gaussian entropy and the relative entropy of non-Gaussianity for these states, plotting the results in Fig.~\ref{fig:NG}. All entropies are nonzero for $\phi \in (0,2\pi)$ , and $NG^{(3)}(\psi_{\phi})$ closely approximates $NG^{(\infty)}(\psi_{\phi})$; this suggests that $NG^{(k)}$ for small $k$ can serve as a good approximation of the relative entropy of non-Gaussianity.

\begin{figure}[t]
  \center{\includegraphics[width=7cm]{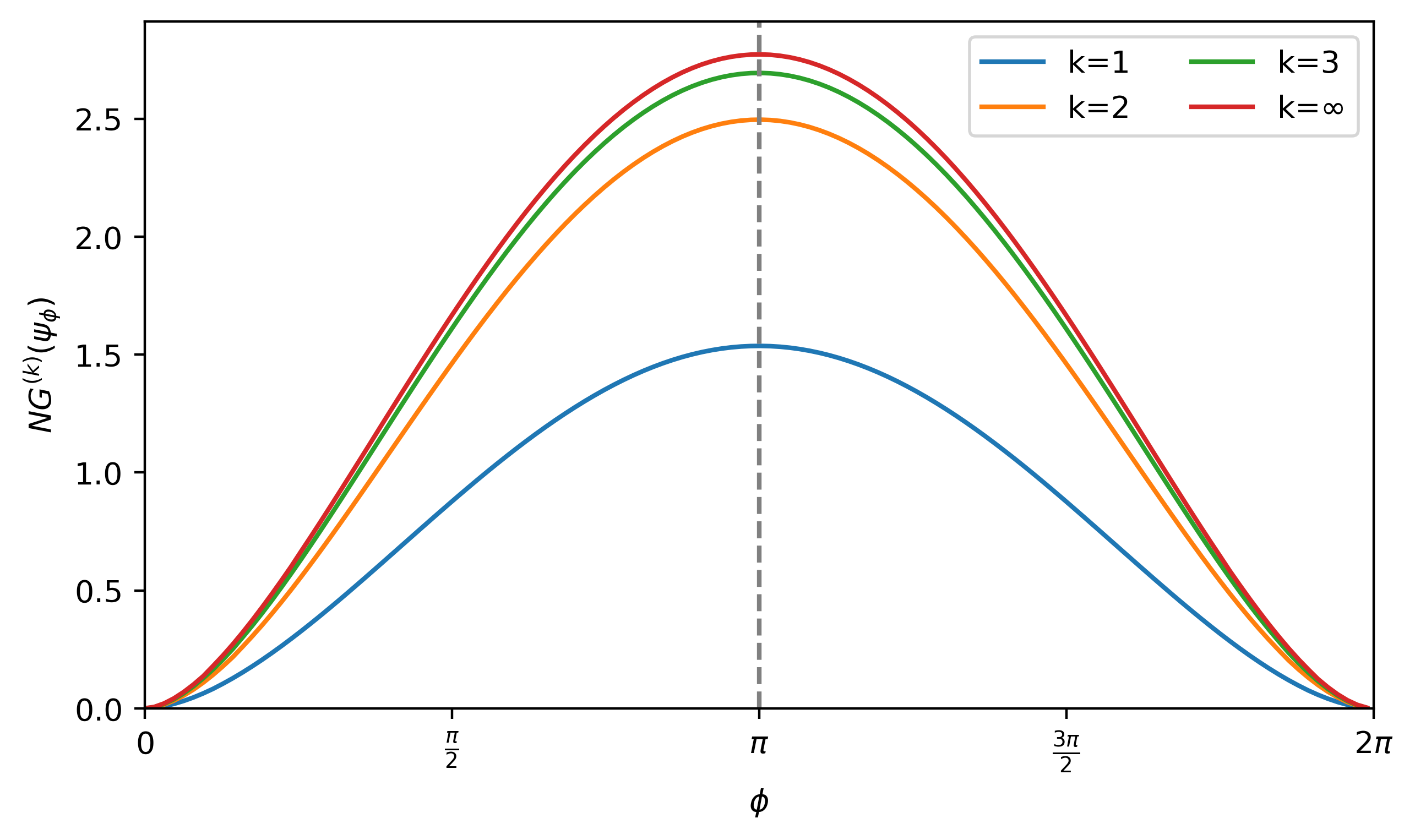}}     
  \caption{non-Gaussian entropy $NG^{(k)}(\psi_\phi)$ of the $4$-qubit state $\psi_\phi$ in~\eqref{def:4QubitMagic}  for $k=1, 2, 3, \infty$.}
  \label{fig:NG}
 \end{figure}

\textit{Remark:} We can also replace the quantum entropy $S(\rho)$ by the quantum R\'enyi entropy $S_{\alpha}(\rho)=\frac{1}{1-\alpha}\log\Tr{\rho^{\alpha}}$, resulting in the 
$k$-th order, $\alpha$ R\'enyi non-Gaussian entropy $NG^{(k)}_{\alpha}(\psi)=S_{\alpha}(\boxtimes^k_{\alpha})(\psi)$; this is also good measure to quantify the fermionic non-Gaussianity (Appendix C). 
To extend the non-Gaussian entropy to mixed states, we have two different forms depending on the choice of the set 
free states $\mca F$ for fermionic quantum computation. For $\mca F=\textbf{Gauss}$ as discussed in the main text, we can define the 
the $k$-th order non-Gaussian entropy as $NG^{(k)}(\rho)=S(\boxtimes^k\rho)-S(\rho)$, which also satisfies the properties in Proposition \ref{prop:NGR}.
If we choose $\mca F$ to be the convex hull of $\textbf{Gauss}$, then we can use the technique of convex roof to define the $k$-th order non-Gaussian entropy for mixed states
$\tilde{NG}^{(k)}(\rho)=\min_{\rho=\sum_i\lambda_i\proj{\psi_i}}\sum_i\lambda_iNG^{(k)}(\psi_i)$, where the minimization is taken over all 
pure state decomposition of $\rho$. It not only satisfies the properties in Proposition \ref{prop:NGR} but also the convexity, i.e., 
$\tilde{NG}^{(k)}(\sum_ip_i\rho_i)\leq \sum_ip_i\tilde{NG}^{(k)}(\rho_i)$.

\section{Conclusion and discussion}
In this work, we investigated 
fermionic convolution and its application to testing Gaussian components and quantifying non-Gaussian resources. 
These results offer a novel perspective on the role of non-Gaussianity 
in universal fermionic quantum computation.

This work also opens exciting avenues for future exploration. 
It will be insightful to relate the Non-Gaussian entropy, an important variational 
information quantity, with measures which quantify the classical simulability of 
fermionic non-Gaussian circuits. 
Additionally, although using only three copies already provides an efficient Gaussian test, 
it would be useful determine whether three copies are indeed optimal. 
Last but not least, 
a series of studies \cite{BGJ23a,BGJ23b,BGJ23c,BJ24a,BGJ24a,BJW24a,bu2024extremality} 
have characterized stabilizer states through discrete convolution, 
highlighting the Gaussian nature of stabilizer operations. 
This suggests potential for a unified framework for classically 
simulable states and circuits, leveraging their common Gaussian structures. 
Such a framework could pave the way for discovering new families of 
classically simulable quantum circuits by identifying new Gaussian operators. 

\section{Acknowledgments} 
We thank Arthur Jaffe, Weichen Gu, and Sasha Berger for helpful discussions. This work is supported in part
by the ARO Grant W911NF-19-1-0302 and the ARO
MURI Grant W911NF-20-1-0082.

\textit{Note:}  During the preparation of this manuscript, we became aware of another work studying fermionic convolution and the central limit theorem~\cite{Coffman24}.

%

\clearpage
\newpage
\onecolumngrid
\starttocentries

\appendix
\tableofcontents
\addtocontents{toc}{\protect\setcounter{tocdepth}{2}}
\addcontentsline{toc}{section}{Appendix}

\bigskip 

Here is an outline of the appendix:
In  Sec.~\ref{appen:algebra},
we introduce the mathematical foundations of 
Clifford and Grassman algebras, paying special attention to the the Fourier transform. 
In Sec.~\ref{Append:GauFourier}, we use the Fourier coefficients to defining moments, cumulants, 
and applythese concepts to characterize fermionic Gaussian states and quantify non-Gaussianity. 
In Sec.~\ref{appen:convolution}, 
we provide detailed proofs for the properties of fermionic convolution. 
Sec.~\ref{app:CLTBound} focuses on the central limit theorem for the fermionic 
convolution and provide a quantitative bound on the rate of convergence.
In Sec.~\ref{appen:Gassification}, we explore the properties of the Gaussification map.
In Sec.~\ref{app:gaussianUnitaryTest}, 
we extend Gaussian testing from states to unitaries.


\section {Background in Clifford and Grassmann algebras}\label{appen:algebra}
We review the mathematical background in the Clifford algebra $\Cl_{2n}$ over generators $\{\gamma_1, \cdots, \gamma_{2n}\}$ and the Grassmann algebra $\G_{2n}$ over generators $\{\eta_1, \cdots, \eta_{2n}\}$. 
For more details, see~\cite{cahill1999density,nielsen2005fermionic,bravyi2004lagrangian}.
\subsection{Clifford and Grassmann algebras}
\label{app:cliffordGrassmannAlgebras}

\begin{definition}[Clifford algebra]
    A finitely-generated Clifford algebra $\Cl_{2n}$ over $2n$ generators $\gamma_1, \cdots, \gamma_{2n}$ consists of complex polynomials over the generators $\gamma$ subject to the anti-commutation relation 
    \[ 
        \{\gamma_j, \gamma_k\} = 2\delta_{jk}, \quad j, k=1, \cdots, 2n. 
    \]
    In this work, we also require the generators to be self-adjoint $\gamma_j^* = \gamma_j$. The algebra $\Cl_{2n}$ is $2^{2n}$-dimensional, with basis elements $ \{\gamma_J:J\subset [2n]\}$ indexed by ordered subsets of $[2n]=\{1, \cdots, 2n\}$ according to 
    \[ 
        \gamma_J = \gamma_{J_1}\gamma_{J_2}\cdots \gamma_{J_{|J|}}. 
    \]
\end{definition}
Note that self-adjointness of the generators is not a standard requirement in mathematics literature, so strictly speaking $\Cl_{2n}$ 
should be denoted the Majorana algebra. However, we will continue referring to \(\Cl_{2n}\) as the Clifford algebra in this context in line with the terminology used in~\cite{hudson1980translation}. 

The Jordan-Wigner transform faithfully represents the Clifford algebra $\Cl_{2n}$ on the space $\Hi_n$ of $n$-qubit operators. The self-adjoint generators $\gamma_j^*=\gamma_j$ are represented as Hermitian Pauli operators, and the generated basis $\{\gamma_J\}$ is orthonormal under the Hilbert-Schmidt inner product 
\begin{eqnarray}
    \la A, B\ra = \df 1 {2^n} \tr(A^\dag B), \quad A, B\in \Hi_n. 
\end{eqnarray}
Using this representation, we identify 
\begin{eqnarray}
    \tr(1_{\Cl_{2n}}) = 2^n\quad \tr(\gamma_{J\neq \emptyset}) = 0.
\end{eqnarray}

\begin{definition}[Grassmann algebra]
    A finitely generated Grassmann algebra $\G_{2n}$ over $2n$ generators $\eta_1, \dots, \eta_{2n}$ 
    consists of complex polynomials over $\eta$ subject to the multiplication rules:
    \leqalign{eq:cliffordRelationApp}{
        \eta_a^2=0, \quad \{\eta_a, \eta_b\} = 0 \text{ when } a\neq b. 
    }
    The Grassmann basis $\{\eta_J\}$ is defined analogously to the Clifford case. 
\end{definition}
We use the orthonormal representation of $2n$ Grassmann generators on a $2n$-qubit Hilbert space $\Hi_{2n}$ as 
\begin{eqnarray}
    \eta_j = Z^{\ot j-1}\ot\sigma^+ \ot I^{\ot 2n-j}, \forall j\in [2n].
\end{eqnarray}
Here $\sigma^+ = \sqrt 2\, |0\ra \la 1|$, which satisfies $Z\sigma^+ = -\sigma^+Z=\sigma^+$. 
Note that the $2n$ generators only span a subspace of $\Hi_{2n}$. 
The generated basis $\{\eta_J\}_{J\subset [2n]}$ is orthonormal under the Hilbert-Schmidt product 
\begin{eqnarray}
    \la A, B\ra = \df 1 {2^{2n}} \tr(A^\dag B). 
\end{eqnarray}
Recall that a $*$-representation of an algebra is one on which the $*$-operator is represented by conjugate-transpose. 
Unlike the Clifford case, there is no faithful $*$-representation of the Grassmann algebra with Hermitian generators $\eta_j^\dag = \eta_j$. To see this, $\eta_j^*=\eta_j$ implies that the $*$-representation of generators must have real eigenvalues, but nilpotency $\eta_j^*\eta_j = 0$ imply that the representation must be trivial. 
In light of this, we cannot directly use the Hilbert-Schmidt formula to define the inner product, which are instead defined as 
\begin{eqnarray}
    \la \eta_J, \eta_K\ra = \df 1 {2^{2n}} (-1)^{|J|(|J|-1)/2} \tr(\eta_J\eta_K), 
    \quad \|A\|_2^2 = \la A, A\ra. 
\end{eqnarray}
This is the Hilbert-Schmidt product upon manually enforcing $\eta_j^\dag = \eta_j$. 
Using this representation, we identify 
\begin{eqnarray}
    \tr(1_{\G_{2n}}) = 2^{2n}, \quad \tr(\eta_{J\neq \emptyset})=0. 
\end{eqnarray}
We also need the anti-commuting tensor product when dealing with transformation between algebras. 

\begin{definition}[Anti-commuting tensor product]
Given two finitely generated algebras $\mca A =\langle\{a_1, \cdots, a_n\}\rangle$ and 
$\mca B=\langle\{b_1, \cdots, b_m\}\rangle$, the anti-commuting tensor product $\mca A\otimes_f \mca B$ is
is generated by
    \begin{eqnarray}
        \{a_j \otimes_f 1:1\leq j \leq n\}\cup \{1\otimes b_k: 1\leq k \leq m\}. 
    \end{eqnarray}
    Multiplication is defined by the following relation,
    \leqalign{eq:antiCommTP}{
        (a_j\otimes_f 1)(1\otimes_f b_k) = -(1\otimes_f b_k)(a_j\otimes_f 1). 
    }
\end{definition}
To simplify notation, we typically equate $a_j, b_k$ with their equivalents in $\mca A\otimes_f \mca B$, leading to the expression $a_jb_k = -b_ka_j$ as shown in equation~\eqref{eq:antiCommTP}. We resort to the full notation only when explicit clarification of spaces is necessary.


\begin{definition}[Even subspace]
    The even subspace (subalgebra) $\mca {A_E}$ of a finitely-generated algebra $\mca A$ is spanned by the products of an even number of generators i.e., all nontrivial terms in the polynomial expansion are of even degree. 
\end{definition}
In this work, we will focus the even space with following isomorphism
\begin{eqnarray}
    \mca A \ot_f \mca {B_E} \cong \mca A \ot \mca {B_E}, 
\end{eqnarray}
where $\mca A$ and $\mca B$ could be either the Clifford algebra $\Cl_{2n}$ or the Grassmann algebra $\G_{2n}$. This identification exploits the property that even elements $b_K$ commute with $a_J$ irrespective of using $\ot$ or $\ot_f$: 
\begin{eqnarray} 
(a_J \ot_f 1)(1\ot b_K) \in \mca A \ot_f \mca {B_E} \leftrightarrow a_J \ot b_K \in \mca A \ot \mca {B_E} 
\end{eqnarray}
The reason we primarily focus on the even subspace is that the desirable mathematical properties of fermion operations are typically defined in $\mca A\ot_f B$ but, operationally, we can only access the algebra $\mca A\ot \mca B$. 


\subsection{Grassmann-Clifford Fourier transform}
\label{app:GrassmannFourierTransform}
The Clifford algebra is intricately related to the Grassmann algebra by 
a Fourier transform which effectively corresponds to a formal 
relabeling of generators. Here, we recall the Grassmann algebra 
and consider this Fourier transform, 
which to the best of our knowledge appeared under the name of ``moment-generating function'' in~\cite{hudson1980translation}.

Recall that the classical Fourier transform is the integral with respect to kernel 
$e^{i\mbf t\cdot \mbf x} $ as follows
\begin{eqnarray}
    \mca F[f](\mbf t) = \int e^{i\mbf t\cdot \mbf x} f(\mbf x)\, d\mbf x. 
\end{eqnarray}
Hence, we consider the following kernel in the definition of the 
Grassman-Clifford Fourier transform,
\begin{align}
    e^F\quad \text{with}~F:=\gamma^T \eta = \sum_j \gamma_j \eta_k \in \Cl_{2n}\otimes_f \G_{2n}.
\end{align}
The kernel expands into the algebra basis as 
\begin{eqnarray}
        \label{eq:FourierKernel}
        e^F = \sum_J (-1)^{j(j-1)/2} \eta_J \gamma_J = \sum_J \eta_J \gamma_J^\dag. 
    \end{eqnarray}
where $j\equiv |J|$ is the size (degree) of the index (monomial). 
To see this, the polynomial expansion of $e^F$ is trivial after the $2n$-th degree. 
In the $k$-th degree (in $\gamma$ or $\eta$), 
the $1/k!$ Taylor coefficient cancels with the $k!$ ways of picking 
nontrivial terms across identical products, and reordering 
$\eta_{J_1}\gamma_{J_1}\cdots \eta_{J_j}\gamma_{J_j}$ into $\eta_J \gamma_J$ 
requires $j(j-1)/2$ swaps.

Denote by $\tr_{\Cl}$ and $\tr_{\G}$ the partial trace over the Clifford and Grassmann algebras, respectively, then 
\leqalign{}{
    \Xi_A(\eta) &= \mca F[A](\eta) = \tr_{\Cl}\left(e^F A\right) 
    = \df 1 {2^n} \sum_J A_J \, \tr_\Cl (\eta_J \gamma_J^\dag \gamma_J) = \sum_J A_J \eta_J,
}
is the Grassman-Clifford Fourier transform. 
The conjugate action  of Gaussian unitaries 
can be rewritten in the Grassmann formalism as 
\leqalign{eq:gaussianUnitaryGrassmannAction}{
      \Xi_{UAU^\dag}(\eta) = \Xi_A(R\eta),
}
where $\Xi_A(R\eta)$ means substituting $\eta_j\mapsto \sum_k R_{jk}\eta_k$ in the expression of $\Xi_A(\eta)$.

\subsection{Even states and unitaries}

Recall that an element $A\in \Cl_{2n}$ is even iff its expansion in the majorana basis only contains nontrivial even-degree terms. In this section, we derive two operational tests which distinguish when a state $\rho$ or unitary $U$ is even. We work in the Jordan-Wigner representation of $\Cl_{2n}$. 

\begin{proposition}
    \label{prp:evenOpCharacterization}
    An operator $A\in \Cl_{2n}$ is even if and only if 
    \begin{eqnarray}
        [A, Z^{\ot n}] = 0. 
    \end{eqnarray}
    Note that $Z^{\ot n} = (-i)^n \gamma_{[2n]}$ is proportional to the parity operator $\gamma_{[2n]}$ on $\Cl_{2n}$. 
\end{proposition}
\begin{proposition}
    \label{prp:evenUnitaryCharacterization}
    A unitary $U\in \Cl_{2n}$ is even if and only if 
    \begin{eqnarray}
        Z^{\ot n}U|+\ra^{\ot n} = UZ^{\ot n}|+\ra^{\ot n}. 
    \end{eqnarray}
    \begin{proof}
        By Proposition~\ref{prp:evenOpCharacterization}, $U$ is even iff $UZ^{\ot n} = Z^{\ot n}U$. This holds if and only if they share the computational eigenbasis $\{|x\ra:x\in \{0, 1\}^n\}$. 
    \end{proof}
\end{proposition}
\begin{proposition}
    \label{prp:evenStateCharacterization}
    A state $\rho$ is even if and only if 
    \begin{eqnarray}
        \rho = Z^{\ot n} \rho Z^{\ot n}. 
    \end{eqnarray}
    \begin{proof}
       This follows from commutativity in Proposition~\ref{prp:evenOpCharacterization}. 
    \end{proof}
\end{proposition}
Note that in this context, a state being ``even'' means that it has definite computational-basis parity, not specifically even computational-basis parity. For example, both $|0\ra$ and $|1\ra$ are even elements of $\Cl_{2n}$, but $|+\ra$ is not. 
\begin{lemma}[Even state test]
    \label{lem:evenStateTest}
    A pure state $\ket{\psi}$ is even iff $(\proj{\psi}, Z^{\ot n} \proj{\psi}Z^{\ot n)}$ passes the swap test. 
\end{lemma}
\begin{lemma}[Even unitary test]
    \label{lem:evenUnitaryTest}
    A unitary $U$ is even iff $(Z^{\ot n}U|+\ra^{\ot n}, UZ^{\ot n}|+\ra^{\ot n})$ passes the swap test. 
\end{lemma}

\section{Characterization of Gaussianity by Fourier coefficients}\label{Append:GauFourier}

In this section, we introduce Fourier coefficients, moments, and cumulants for the fermionic system. These quantities are used to characterize fermionic Gaussian operators and quantify non-Gaussianity.
\subsection{Moments and cumulants}
\label{app:momentCumulant}
The moments and cumulants of a state $\rho\in \Cl_{2n}$ are defined analogous to their classical statistics counterparts: moments are overlaps with basis projective operators, while cumulants are specific combinations of moments that are additive under convolution.
 
\begin{definition}[Moments]
\label{def:moment}
Given a multi-index $J\subset [2n]$, the $J$-th moment of a $n$-qubit state $\rho\in \Cl_{2n}$ is 
\begin{align}
    \rho_J=\tr(\gamma^\dag_J\rho).
\end{align}
The moment-generating operator (or function) $\Xi_\rho\in \G_{2n}$ is the Fourier transform of $\rho$: 
\begin{eqnarray}
    \Xi_\rho(\eta) = \sum_J \rho_J \eta_J.
\end{eqnarray}
Note that the orthonormality of $\{\gamma_J\}$ also implies 
\leqalign{eq:momentExpansion}{
    \rho = \df 1 {2^n} \sum_J \rho_J \gamma_J \in \Cl_{2n}.
}
\end{definition}
Any state is uniquely determined by its moments $\set{\rho_J}$. 
Since $i^{|J|}\gamma_J$ is a Hermitian with eigenvalues $\pm 1$, 
$ |\rho_J|=|\Tr{\gamma^\dag_J\rho}|\leq 1$.
The orthonormal moment expansion in equation~\eqref{eq:momentExpansion} also implies 
\begin{align*}
    \sum_{J\subset [2n]}|\rho_J|^2
    = 2^n \Tr{\rho^2}.
\end{align*}

We define the degree-dependent moment weight, which we show to be a Gaussian unitary-invariant. 
\begin{definition}[Moment weight]
Given a state $\rho$, the $k$-th moment weight is 
\begin{eqnarray}
    \label{app:weightNorms}
   W_k(\rho) = \sum_{|J|=k} |\rho_J|^2.
\end{eqnarray}
And the total moment  weight is
\begin{align}
    I_M(\rho)=\sum^{2n}_{k=0}j  W_k(\rho). 
\end{align}
\end{definition}
The total moment weight is also called sensitivity~\cite{Bu2024complexity}, which 
has been shown to be useful in the study of quantum circuit complexity.

\begin{proposition}
    \label{prp:degreeWeightInvariance}
    The $k$-th moment weight is Gaussian unitary-invariant: for every integer $0\leq k\leq 2n$: 
    \begin{align}
             W_k(U_G\rho U^\dag_G)= W_k(\rho). 
    \end{align}
    Consequently, the total moment weight is also preserved,
    \begin{align}
       I_M(U_G\rho U^\dag_G)=I_M(\rho). 
    \end{align}
\end{proposition}
\begin{proof}
      Fixing $0\leq k\leq 2n$, consider the $k$-th order tensor $(\Sigma^{(k)}_{\rho})_{l_1,...,l_k}$ with $l_i\in [2n]$, 
        \[
            (\Sigma^{(k)}_{\rho})_{l_1,...,l_k}= 
            \begin{cases}
                \rho_{l_1,..,l_k}, & \text{if } l_1,..,l_k \text{ are all different},\\
                0, & \text{otherwise.}
            \end{cases}
        \]
        It is easy to see that $\norm{\Sigma^{(k)}_{\rho}}^2_2:=\sum_{l_1,..,l_k}|(\Sigma^{(k)}_{\rho})_{l_1,...,l_k}|^2=W_k(\rho)$. We need 
        to show that $\norm{\Sigma^{(k)}_{\rho}}^2_2=\norm{\Sigma^{(k)}_{U_G\rho U^\dag_G}}^2_2$ for any Gaussian unitary $U_G$ which corresponds to a generator rotation $R$ such that 
        $U_G\gamma U_G^\dag= R\gamma$. Then 
        \begin{align}
           (\Sigma^{(k)}_{U_G\rho U^\dag_G})_{l_1,...,l_k}
            =\sum_{j_1,..,j_k}\rho_{j_1...,j_k}\prod^k_{i}R_{j_i,l_i}
            =\sum_{j_1,..,j_k}(\Sigma^{(k)}_{\rho})_{l_1,...,l_k}M_{(j_1,...,j_k),(l_1,...,l_k)}, 
        \end{align}
        where $M_{(j_1,...,j_k),(l_1,...,l_k)}=\prod^k_{i}R_{j_i,l_i}\cong R^{\ot k}$ is a $(2n)^k\times (2n)^k$ rotation matrix since $R$ is orthogonal. Thus, the weights of $\Xi_\rho(\eta)$ by each degree are equal to the weights of $\Xi_{U\rho U^\dag}(\eta) = \Xi_\rho(R\eta)$ by degree since the degree-$k$ weight vector (tensor) $\Sigma^{(k)}_\rho$ undergoes rotation by $R^{\ot k}$ under the unitary algebra automorphism $\eta\mapsto R\eta$. 
\end{proof}

\begin{definition}[Cumulants]
Given a state $\rho$,
     the cumulant-generating operator (or function) $\Psi_\rho \in \G_{2n}$ is 
\leqalign{def:cumulant}{
    \Psi_\rho(\eta) = \log\,  \Xi_\rho(\eta) = \sum_{k=1}^{2n} (-1)^{k+1} \df{(\Xi_\rho - 1_{\G_{2n}})^k}{k} = \sum_J \kappa^\rho_J \eta_J, 
}
where $\kappa^\rho_J$ is the $J$-th cumulant of $\rho$. 
\end{definition}
The cumulant-generating element of a quantum state always exists because $\Xi_\rho - 1_{\G}$ is nilpotent with degree at most $2n$. 
Additionally, note that by definition, the quadratic cumulants and moments concur for even $\rho$: 
\begin{eqnarray}\label{eq:equaiv_2}
    \rho_J=\kappa^\rho_J, \quad\forall J\subset [2n] ~~\text{with} ~~|J|=2.
\end{eqnarray}

Similar to the moment weight, we can also define the cumulant weight. 
\begin{definition}
\label{def:cumulantWeights}
Given a state $\rho$, the $j$-th cumulant  weight is 
    \begin{eqnarray}
    \label{app:CumulantweightNorms}
    K_j(\rho)= \sum_{|J|=j} |\kappa^\rho_J|^2 .
\end{eqnarray}
And the total cumulant weight is 
\begin{align}
    K(\rho)=\sum_{j}j K_j(\rho).
\end{align}
\end{definition}
Based on the equivalence of the second order moments and cumulants in \eqref{eq:equaiv_2}, we have the equivalence of the corresponding weight $W_2(\rho)=K_2(\rho)$ for even states $\rho$. 
\begin{proposition}
    \label{prp:cumulantTpAdditivity}
Given two even states $\rho$ and $\sigma$, we have 
    \begin{eqnarray}
        \Psi_{\rho \otimes \sigma} = \Psi_\rho \otimes 1 + 1 \otimes \Psi_\sigma. 
    \end{eqnarray}
Thus, for even states $\rho$ and $\sigma$, the cumulant weight is additive, i.e.,
    \begin{eqnarray}
     K_j(\rho\ot \sigma)=K_j(\rho)+K_j(\sigma).
    \end{eqnarray}
   
\end{proposition}
 \begin{proof}
 Because
        $\Xi_{\rho \otimes \sigma} = \Xi_\rho \otimes \Xi_\sigma$, we have 
        $\Psi_{\rho \otimes \sigma} = \log (\Xi_\rho \otimes \Xi_\sigma) = \log \Xi_\rho \otimes 1 + 1\otimes \log \Xi_\sigma = \Psi_\rho\otimes 1 + 1 \otimes \Psi_\sigma$. 
    \end{proof}
The proof above can be relaxed to only requiring $\sigma$ to be even; we require both inputs to be even for simplicity.

\begin{proposition}
    \label{prp:degreeWeightInvariance_cu}
    The $j$-th cumulant weight is Gaussian unitary-invariant, i.e., 
    \begin{align}
        K_j(U_G\rho U^\dag_G)=K_j(\rho).
    \end{align}
    Thus the total cumulant weight is also preserved, 
    \begin{align}
        K(U_G\rho U^\dag_G)=K(\rho).
    \end{align}
    
\end{proposition}
\begin{proof}
       The proof is similar to that of Proposition~\ref{prp:degreeWeightInvariance}.
        For cumulants, we have 
        \[
            \Psi_{U\rho U^\dag}(\eta) = \log[\Xi_{U\rho U^\dag}(\eta)] 
            = \log[\Xi_\rho(R\eta)] = (\log \Xi_\rho)(R\eta) = \Psi_\rho(R\eta). 
        \]
        This is because $\eta \mapsto R\eta$ is an automorphism of the algebra, and algebra automorphisms commutes with all algebraic operations defined using addition, scalar multiplication, and algebra multiplication. 
        Here  $\varphi$ is an automorphism means that  $\varphi(ab)=\varphi(a)\varphi(b)$. 
        In particular, the Grassmann logarithm of the Fourier transform of states, defined as a finite power series ~\eqref{def:cumulant}, is an algebraic operation. Applying the same argument to $\Psi_\rho(R\eta)$ demonstrates that the weights of $\Psi_{U\rho U^\dag}(\eta) = \Psi_\rho(R\eta)$ and $\Psi_\rho(\eta)$, when measured by degree, are identical. 
    \end{proof}

\subsection{Gaussian states and unitaries}
Recall that Gaussian states are the ground or thermal states of quadratic Hamiltonians, that is,  
\begin{eqnarray}
    \rho=C\exp\left(\frac{i}{2} \gamma^T h \gamma\right),
\end{eqnarray}
where $(h_{jk})$ is a real, antisymmetric matrix and $C$ is a normalization constant. We will first show that the Fourier transform of the Gaussian state is also Gaussian. 
The covariance matrix of any state $\rho$ is defined as 
\begin{eqnarray}
    (\Sigma_\rho)_{jk}
    = \frac{i}{2}\Tr{\rho[\gamma_j, \gamma_k]} 
    = -i \tr[(\gamma_j\gamma_k)^\dag \rho].
\end{eqnarray}
Since $h$ is a real, antisymmetric matrix, by a standard result in the matrix theory~\cite{horn2012matrix}, 
there exists a rotation $R\in SO(2n)$ such that 
\leqalign{eq:diagonalAntisymmetricMatrix}{
    R^Th R = h_D = \bigoplus_{j=1}^n \begin{bmatrix}
        0 & \nu_j \\ -\nu_j & 0 
    \end{bmatrix}, \quad \nu_j\geq 0. 
}
Using a Gaussian unitary $U_G$ to implement the rotation, we can diagonalize $\rho$ according to
\begin{eqnarray}
\rho_D:=U_G\rho U^\dag_G = C\exp\left[\frac{i}{2} \gamma^T (R^ThR) \gamma\right]
=C\exp\left(\frac{i}{2} \gamma^T h_D\gamma\right).
\end{eqnarray}
The exponential is decomposed as follows,
\leqalign{eq:separableDiagonalGaussian}{
    \exp\left(\df i 2 \gamma^T h_D \gamma\right) 
    &= \exp \left(
        i\sum_{j=1}^n \nu_j \gamma_{2j-1}\gamma_{2j}
    \right)
    = \bigotimes_{j=1}^{n} \exp(i \nu_j \gamma_1 \gamma_2) 
    = \bigotimes \exp(-\nu_j Z) 
    = \bigotimes_{j=1}^n \begin{bmatrix}
        e^{-\nu_j} & 0 \\ 0 & e^{\nu_j} 
    \end{bmatrix}.
}
The normalized expression for a diagonalized Gaussian state is thus 
\leqalign{}{
    \rho_D 
    &= C\exp\left(\df i 2 \gamma^T h_D \gamma\right)  
    = \bigotimes_{j=1}^n \df 1 {e^{-\nu_j} + e^{\nu_j}} \begin{bmatrix}
        e^{-\nu_j} & 0 \\ 0 & e^{\nu_j}
    \end{bmatrix} \\ 
    &= \df 1 {2^n} \bigotimes_{j=1}^n \left( I - \tanh(\nu_j) Z \right)
    = \df 1 {2^n} \bigotimes_{j=1}^n(I + i\tanh(\nu_j) \gamma_1\gamma_2).  
}
Hence, the Fourier transform is 
\leqalign{}{
    \Xi_{\rho_D}(\eta) 
    = \bigotimes_{j=1}^n 1 + i\lambda_j \eta_1\eta_2 
    = \exp \left(\sum_{j=1}^n \lambda_j \eta_{2j-1}\eta_{2j}\right) 
    = \exp \left(\df 1 2 \eta^T \Sigma_{\rho_D} \eta\right),
}
where the covariance matrix takes on the block-diagonal form 
\leqalign{}{
    \Sigma_{\rho_D} = \bigoplus_{j=1}^n \begin{pmatrix}
        0 & \lambda_j \\ -\lambda_j & 0 
    \end{pmatrix}, \quad \lambda_j = \tanh(\nu_j). 
}

Based on the Fourier transform, we obtain the following lemma which can also be found in~\cite{bravyi2004lagrangian} and~\cite{surace2022fermionic}.
\begin{lemma}[Diagonalization of Gaussian states]
    \label{thm:gaussianStateCharacterization}
    Given a Gaussian state 
    \begin{eqnarray}
        \rho = C \exp \left(\df i 2 \gamma^T h \gamma\right), \quad h_{jk} = -h_{kj} \in \R, \quad C\in \R 
    \end{eqnarray}
    such that $h=Rh_DR^T$, where $h_D$ is of the form~\eqref{eq:diagonalAntisymmetricMatrix}, 
    then the Fourier transform of $\rho$ is 
    \leqalign{eq:GaussianFourierExpression}{
        \Xi_\rho = \exp\left(\df 1 2 \eta^T\Sigma_\rho \eta\right),
    }
    where 
     \leqalign{eq:GaussianFourierCovariance}{ 
        \Sigma_\rho = R \left(\bigoplus_{j=1}^n \begin{bmatrix}
            0 & \lambda_j \\ -\lambda_j & 0 
        \end{bmatrix}\right) R^T, \quad \lambda_j = \tanh(\nu_j). 
    }
    The covariance matrix satisfies $0\leq \Sigma_\rho^T\Sigma_\rho \leq I_{2n}$ diagonal, with equality if and only if $\rho$ is pure. The entropy is 
    \leqalign{eq:GaussianDiagonalEntropy}{
        S(\rho) = \sum_{j=1}^n h\left(\df{1 + \lambda_j}{2}\right), \quad h(x) = -x\log x - (1-x)\log(1-x). 
    }
 
\end{lemma}
Thus, the covariance matrix can be used to determine the purity of a Gaussian state. 
\begin{lemma}[\cite{bravyi2004lagrangian, de2013power}]
A Gaussian state $\rho$ is pure if and only the covariance matrix $\Sigma_{\rho}$
satisfies $\Sigma_{\rho}^T\Sigma_\rho=I$.
\end{lemma}

The diagonalization lemma implies that entropy can be traded for higher Gaussian weight. 
\begin{lemma}
    \label{lem:entropyQuadraticweightInequality}
    Given any $n$-qubit Gaussian state $\rho$ with quadratic weight 
    $\rho_{(2)} < n$, for every $0\leq S_0<S(\rho)$ 
    there exists a $n$-qubit Gaussian state $\sigma$ such that 
    \[ 
        S(\sigma) = S_0, \quad   W_2(\sigma)>W_2(\rho). 
    \]
    \begin{proof}
        By the covariance matrix condition in 
      Lemma~\ref{thm:gaussianStateCharacterization} and equivalence between $l_2$ norm of covariance matrix and moment weight, we have
        \malign{
         W_2(\rho)
            = \df 1 2 \|\Sigma_\rho\|_2^2 \leq \df 1 2 \|I_{2n}\|^2 = n,
        }
        with equality iff $\rho$ is pure. Let $\lambda = (\lambda_1, \cdots, \lambda_n)\in [-1, 1]^{n}$ 
        be the imaginary components of the eigenvalues of $\Sigma_\rho$ in~\ref{eq:GaussianFourierExpression}, 
        consider the map which parameterizes the entropy~\ref{eq:GaussianDiagonalEntropy} 
        in terms of $\lambda_{(2)} = (\lambda_1^2, \cdots, \lambda_n^2)$ 
        \malign{
            \tilde S(\lambda_{(2)}) = \sum_{j=1}^n h \left(\df{1+\sqrt{\lambda_j^2}}{2}\right) = S(\rho). 
        }
        Here $h$ is the binary entropy, and we have used the symmetry of~\ref{eq:GaussianDiagonalEntropy} under $\lambda_j\mapsto -\lambda_j$. Then $\tilde S(\lambda_{(2)})$ is monotonically decreasing in $\lambda_j^2\in [0, 1]$ for each $\lambda_j^2$, while by equivalence between $l_2$ norm of covariance matrix and moment weight
        the quadratic weight is strictly increasing in $\lambda_j^2$: 
        \[ 
            W_2(\rho) = \df 1 2 \sum_{j, k=1}^{2n} |(\Sigma_\rho)_{jk}|^2 = \sum_j |\lambda_j|^2. 
        \]
        Thus, for any $0\leq S_0<S(\rho)$, we can find $\lambda'$ for $\sigma$ such that $\tilde S(\lambda_{(2)}')=S_0<S(\rho)$ and $\lambda'_{(2)}\geq \lambda_{(2)}$ component-wise, with inequality strict on at least one component (which implies $W_2(\sigma)>W_2(\rho)$). 
    \end{proof}
\end{lemma}
As a consequence of the Fourier representation of Gaussian states, we also have Wick's theorem (\cite{bravyi2004lagrangian}, equation 17; \cite{surace2022fermionic}, section 3.4). 
First introduce the notation $M_{|J}$: 
Given an  $n\times n$ matrix $M$, $M_{|J}$ denotes the restriction of $M$ onto the subspaces indexed by $J$. For example
    \begin{eqnarray}
        \begin{pmatrix}
            0 & 1 & 2 \\ 3 & 4 & 5 \\ 6 & 7 & 8 
        \end{pmatrix}_{|\{1, 3\}} = \begin{pmatrix}
            0 & 2 \\ 6 & 8 
        \end{pmatrix}, 
    \end{eqnarray}

\begin{lemma}[Wick's theorem]
    \label{cor:wick}
    The higher-order moments of an even Gaussian state $\rho$ with covariance $\Sigma_\rho$ satisfy 
    \begin{eqnarray}
        \rho_J = \Pf[(\Sigma_\rho)_{|J}]
    \end{eqnarray}
    where $\Pf(M)$ is the Pfaffian of a $2n\times 2n$ antisymmetric matrix $M$ defined by 
    \begin{eqnarray}
        \Pf(M) = \df 1 {2^n n!} \sum_{\sigma \in S_{2n}} \mrm{sgn}(\sigma) M_{\sigma_1\sigma_2}\cdots M_{\sigma_{2n-1}\sigma_{2n}}. 
    \end{eqnarray}
\end{lemma}
    \begin{proof}
     Due to Lemma~\ref{thm:gaussianStateCharacterization}, the Fourier transform of $\rho$ is 
        \malign{
            \Xi_\rho 
            &= \exp \left(\df 1 2 \eta^T \Sigma_\rho \eta\right) 
            = \exp \left(\df 1 2 \sum_{j, k=1}^{2n} (\Sigma_\rho)_{jk}\eta_j\eta_k \right) \\ 
            &= 1_{\G_{2n}} + \sum_{m=1}^{2n} \df 1 {2^m m!} \left[\sum_{j, k=1}^{2n} (\Sigma_\rho)_{jk}\eta_j\eta_k \right]^m. 
        }
        Consider the expansion term for $\eta_J$, where $|J|=2m$. We need to construct $\eta_J$ from the $m$-th order above. Each way of picking corresponds to a permutation $\sigma\in S_{2n}$ of $J$, with the sign of rearranging $\eta$'s into $\eta_J$ equal to $\mrm{sgn}(\sigma)$, then 
        \[ 
            \rho_J = \df 1 {2^m m!} \sum_{\sigma\in S_{2n}} (\Sigma_\rho)_{\sigma(j), \sigma(k)} = \Pf[(\Sigma_\rho)_{|J}]. 
        \]
        
    \end{proof}

\subsection{Non-Gaussian measures by Fourier coefficients}
\label{app:magicMeasure}
\begin{figure}[t]
  \centering
  \begin{minipage}{0.48\textwidth}
    \centering
    \includegraphics[width=1\linewidth]{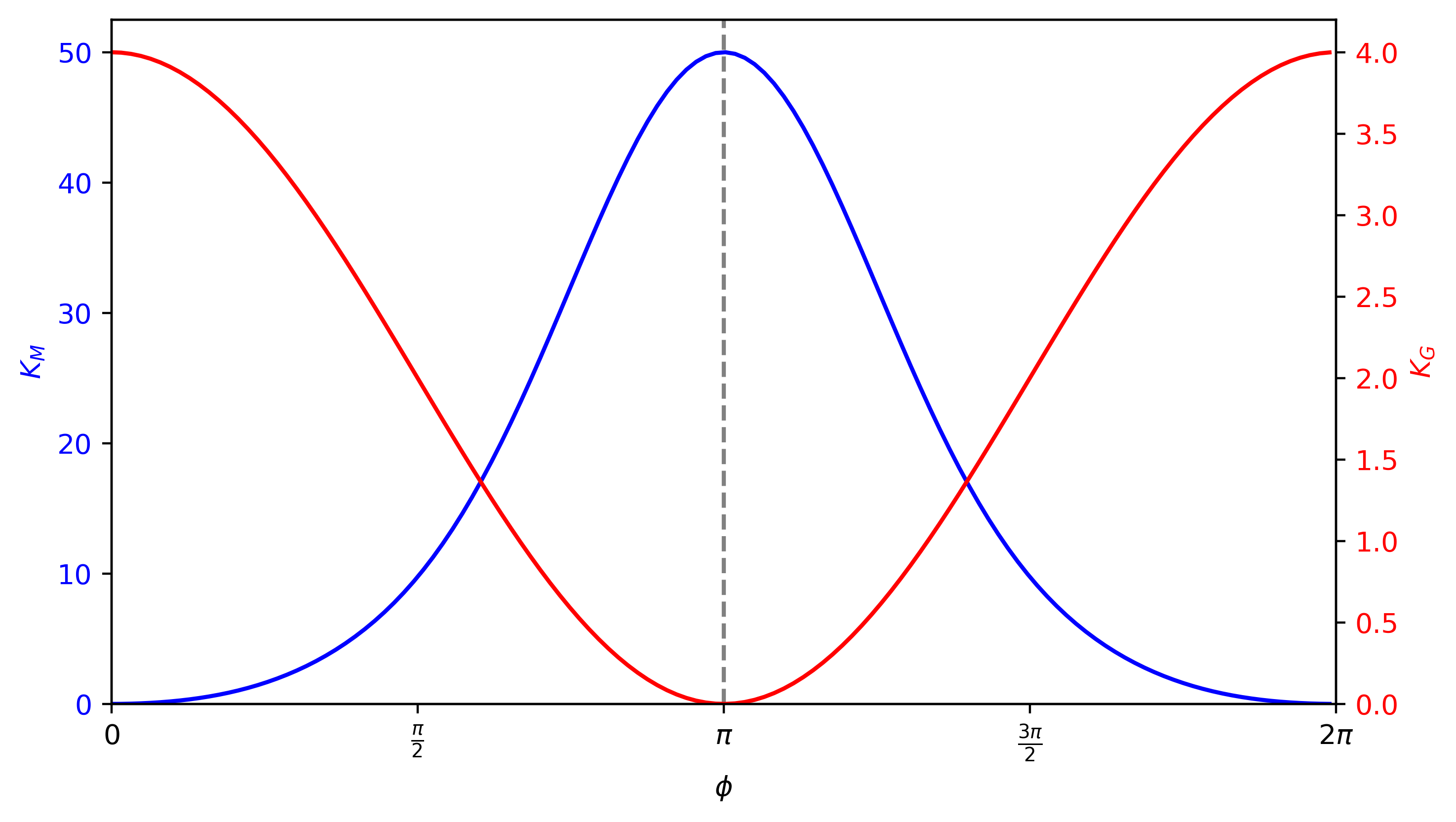}
    \caption{Non-Gaussian $K_M(\rho)$ \textcolor{blue}{(blue)} and Gaussian weight $K_G(\rho)$ \textcolor{red}{(red)} of the $4$-qubit family~\ref{eq:4qubitFamily}. 
    }
    \label{fig:mcn_plot}
  \end{minipage} 
  \begin{minipage}{0.45\textwidth}
    \centering
    \includegraphics[width=1\linewidth]{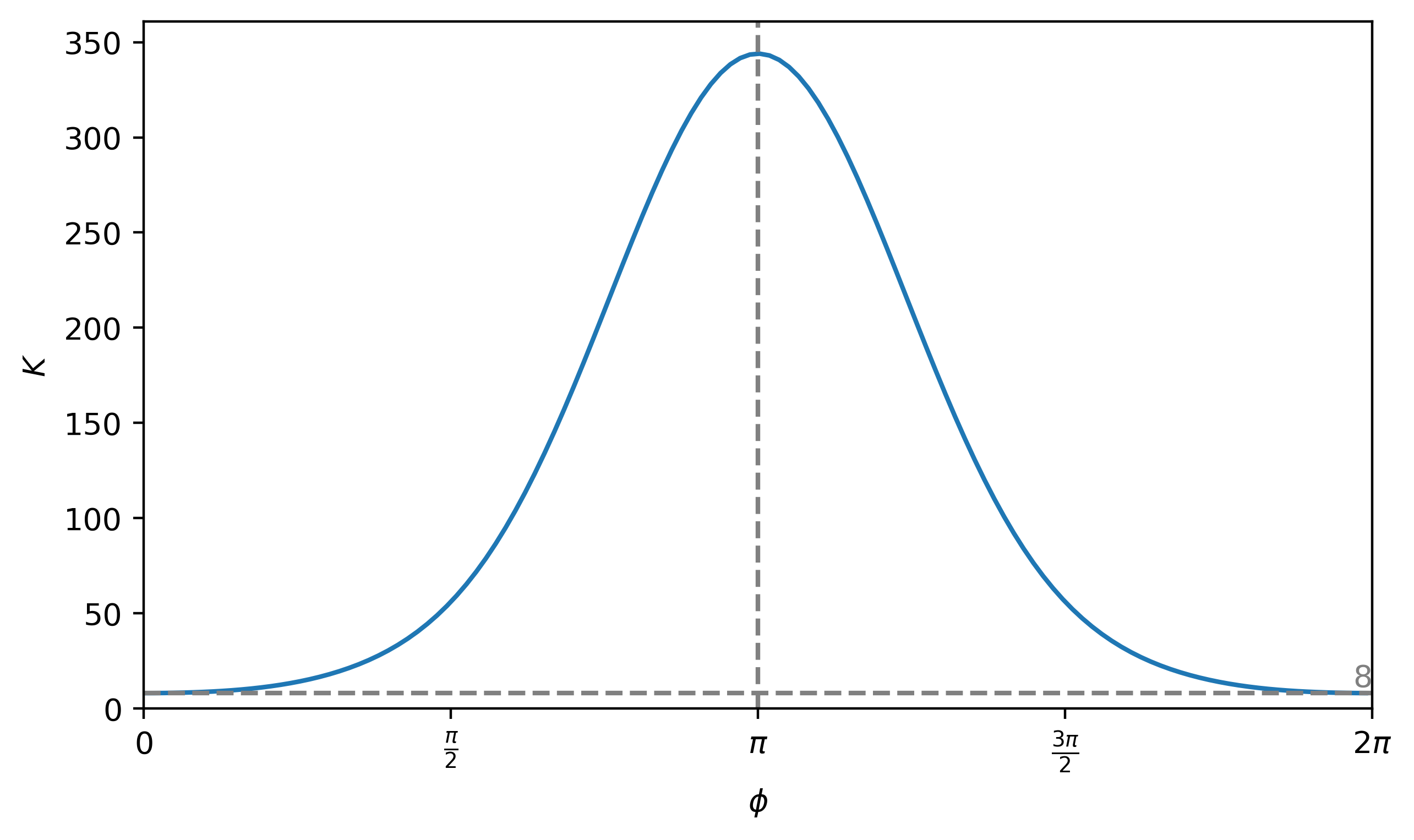}
    \caption{Total cumulant weight $K(\rho)$ (Definition~\ref{def:cumulantWeights}) of the $4$-qubit family~\ref{eq:4qubitFamily}. }
    \label{fig:mcn_weight}
  \end{minipage}
\end{figure}

The definition of the cumulants $\kappa_J$ imply that they can determine the 
the quantum states uniquely. 
The Fourier formula~\eqref{eq:GaussianFourierExpression} of Gaussian states 
show that Gaussian states have vanishing super-quadratic cumulants: 
$\kappa_J=0$ for all $J$ with $|J|\geq 4$. 
Based on this observation, we introduce the 
Gaussian weight and non-Gaussian weight; they are 
derived quantities from the cumulant weights. 

\begin{definition}
    \label{def:GaussianWeights}
    Given an even state $\rho$,   the Gaussian weight is
    \begin{align}
         K_G(\rho) = K_2(\rho),
    \end{align}
    and the non-Gaussian weight is 
    \begin{align}
        K_M(\rho)=\sum^{2n}_{j=4}K_j(\rho).
    \end{align}
\end{definition}

Our previous results have established that the Gaussian weight and non-Gaussian weight can be used to quantify the non-Gaussianity of quantum states, as they fulfill the fundamental properties required of a resource measure: 
\begin{proposition}
    The Gaussian weight $K_G$
    satifies the following properties: 
       \begin{enumerate}[(1)]
        \item Extremality: 
        Among all states with the given entropy $0\leq S_0 \leq n\log 2$, $K_G$ is maximized by Gaussian states. 
        \item Gaussian-invariance: $K_G(U_G\rho U^\dag_G)=K_G(\rho)$ for any Gaussian unitary $U_G$.
        \item Additivity under tensor product: $ K_G(\rho\ot \sigma)
        =K_G(\rho)+K_G(\sigma)$ for any even states $\rho$ and $\sigma$.
    \end{enumerate}
\end{proposition}
  \begin{proof}
        Extremality is established in Proposition~\ref{prp:gaussianMaxQuadraticWeight}. 
        Gaussian-invariance follows from Proposition~\ref{prp:degreeWeightInvariance}. 
        Additivity under tensor product follows the additivity of cumulants in Proposition~\ref{prp:cumulantTpAdditivity}. 
    \end{proof}

\begin{proposition}
    The non-Gaussian weight $K_M$ satisfies: 
    \begin{enumerate}[(1)]
        \item Extremality: $K_M(\rho)\geq 0$, with equality iff $ \rho$ is Gaussian. 
        \item Gaussian-invariance: $K_M(U\rho U^\dag) = K_M(\rho)$. 
        \item Tensor-product additivity: $K_M(\rho \ot \sigma) = K_M(\rho) + K_M(\sigma)$.
    \end{enumerate}
 
\end{proposition}
   \begin{proof}
        Extremality follows from the CLT convergence bound~\ref{app:convergenceBound}. 
       Due to the definition of $K_M$, Gaussian-invariance and additivity follows from Propositions~\ref{prp:degreeWeightInvariance} and~\ref{prp:cumulantTpAdditivity}. 
    \end{proof}

The total cumulant weight can also be used as a non-Gaussian measure: 
\begin{proposition}
    The total weight satisfies the following properties for pure state $\psi$,
     \begin{enumerate}[(1)]
        \item Faithfulness: $K(\psi)\geq 2n$, with equality iff $\psi$ is Gaussian. 
        \item Gaussian-invariance: $K(U_G\proj{\psi} U^\dag) = K(\proj{\psi})$. 
        \item Tensor-product additivity: $K(\psi_1 \ot \psi_2) = K(\psi_1) + K(\psi_2)$.
    \end{enumerate}
 
\end{proposition}
\begin{proof}
Faithfulness is established by $K\geq 2K_G$ with equality iff $K_G=n$ iff $\rho$ is pure Gaussian. 
Due to the definition of total weight, Gaussian-invariance and additivity follows from Propositions~\ref{prp:degreeWeightInvariance} and~\ref{prp:cumulantTpAdditivity}. 
\end{proof}

As an example, consider the resource measures Gaussian weight $K_G$, non-Gaussian weight $K_M$, and total weight $K$ for the $4$-qubit  parametrized state $\psi_\phi$, 
\leqalign{eq:4qubitFamily}{
    \ket{\psi_{\phi}}=\frac{1}{2}
    \left(\ket{0000}+\ket{0011}+\ket{1100}+e^{i\phi}\ket{1111} \right),
}
with $\phi\in [0,2\pi]$. The numerical results are shown 
in Figure~\ref{fig:mcn_weight}. We numerically find that  $\psi_\pi$ is maximally non-Gaussian among $4$-qubit states according to all three measures.

\section{Fermionic convolution: definition and properties}\label{appen:convolution}
We begin by giving a circuit decomposition of the convolution unitary under 
the standard Jordan-Wigner transform in~\ref{app:convolutionUnitary}. 
The properties of fermionic convolution are enumerated and proven in~\ref{Appen:prop_con}. 

\subsection{Convolution unitary}
\label{app:convolutionUnitary}
From here on, we refer to the fermionic beam splitter as the convolution unitary 
\begin{eqnarray}
    W_\theta = \exp \left(\df \theta 2 \sum_{j=1}^{2n} \gamma_j \gamma_{2n+j}\right).
\end{eqnarray}
When $\theta=\pi/4$, the fermionic beam splitter has the following balanced form
\leqalign{eq:convUnitaryGeneratorAction}{
    W_{\pi/4}\gamma_j W^\dag_{\pi/4}=\frac{1}{\sqrt{2}} \gamma_j-\frac{1}{\sqrt{2}}\gamma_{2n+j},\\
    W_{\pi/4}\gamma_{2n+j} W^\dag_{\pi/4}=\frac{1}{\sqrt{2}}  \gamma_j+\frac{1}{\sqrt{2}}  \gamma_{2n+j}.
}
Consider the simplest case, 
where $W^{(1)}_\theta$ denotes the convolution unitary which acts on two $1$-qubit inputs: 
\begin{eqnarray}
    W^{(1)}_{\theta}=\exp\left[ i\frac{\theta}{2}(XY+YX)\right]. 
\end{eqnarray}
The decomposition of $W_\theta^{(1)}$ in terms of elementary gates is shown in Fig~\ref{fig:elementary_decomp}. To decompose the multi-qubit convolution unitary, we conjugate $W^{(1)}_\theta$ with a swapping gadget (Fig~\ref{fig:swap_gadget}) which swaps the fermionic operators of two lines. Fig~\ref{fig:total_decomp} demonstrates the $4$-qubit case with two inputs on lines $0$ to $3$, $4$ to $7$, respectively. Circuits are visualized using Qiskit~\cite{wille2019ibm}. 

The $1$-qubit gates in Fig~\ref{fig:elementary_decomp} are 
\begin{eqnarray}
    X = \begin{pmatrix}
        0 & 1 \\ 1 & 0
    \end{pmatrix}, \quad 
    H = \df 1 {\sqrt 2} \begin{pmatrix}
        1 & 1 \\ -1 & 1
    \end{pmatrix}, \quad 
    S = \begin{pmatrix}
        1 & 0 \\ 0 & -i
    \end{pmatrix}. 
\end{eqnarray}
In the computational basis $|00\ra, |01\ra, |10\ra, |11\ra$, the $2$-qubit gates 
in Fig~\ref{fig:elementary_decomp} and Fig~\ref{fig:swap_gadget} are 
\begin{eqnarray}
    R_Z(\theta)_{1\to 2} 
    = \begin{pmatrix}
        1 & 0 & 0 & 0 \\ 0 & 1 & 0 & 0 \\ 
        0 & 0 & 1 & 0 \\ 0 & 0 & 0 & e^{i\theta} 
    \end{pmatrix}, \quad 
    \mrm{CZ} = \begin{pmatrix}
        1 & 0 & 0 & 0 \\ 0 & 1 & 0 & 0 \\ 
        0 & 0 & 1 & 0 \\ 0 & 0 & 0 & -1 
    \end{pmatrix}, \quad 
    \mrm{SWAP} = \begin{pmatrix}
        1 & 0 & 0 & 0 \\ 
        0 & 0 & 1 & 0 \\ 
        0 & 1 & 0 & 0 \\ 
        0 & 0 & 0 & 1 
    \end{pmatrix}. 
\end{eqnarray}

\begin{figure}[h]
    \centering
    \includegraphics[width=1\textwidth]{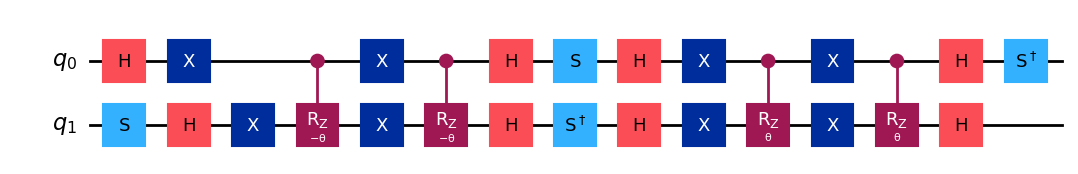}
    \caption{Elementary decomposition of the $1$-qubit convolution unitary $W_\theta^{(1)}$. 
    Here $R_z(\theta) = \exp(-i\theta Z/2)$.}
    \label{fig:elementary_decomp}
\end{figure}
\begin{figure}[h]
    \centering
    \includegraphics[width=1\textwidth]{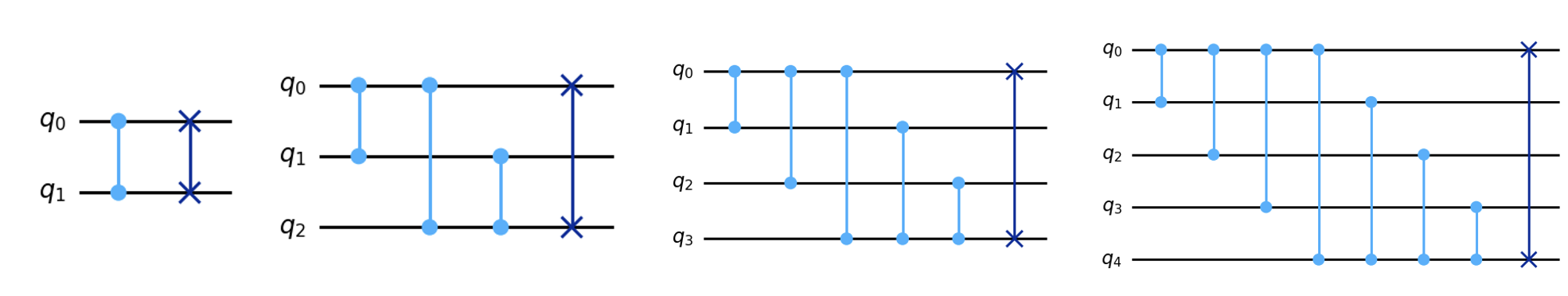}
    \caption{Swap gadget for $2, 3, 4, 5$-qubit convolution unitaries. 
    The blue and navy gates are $CZ$ and $\mrm{SWAP}$, respectively. 
    This gadget swaps the two majorana operators associated with the top and bottom lines 
    while leaving others unchanged. }
    \label{fig:swap_gadget}
\end{figure}
\begin{figure}[h]
    \centering
    \includegraphics[width=1\textwidth]{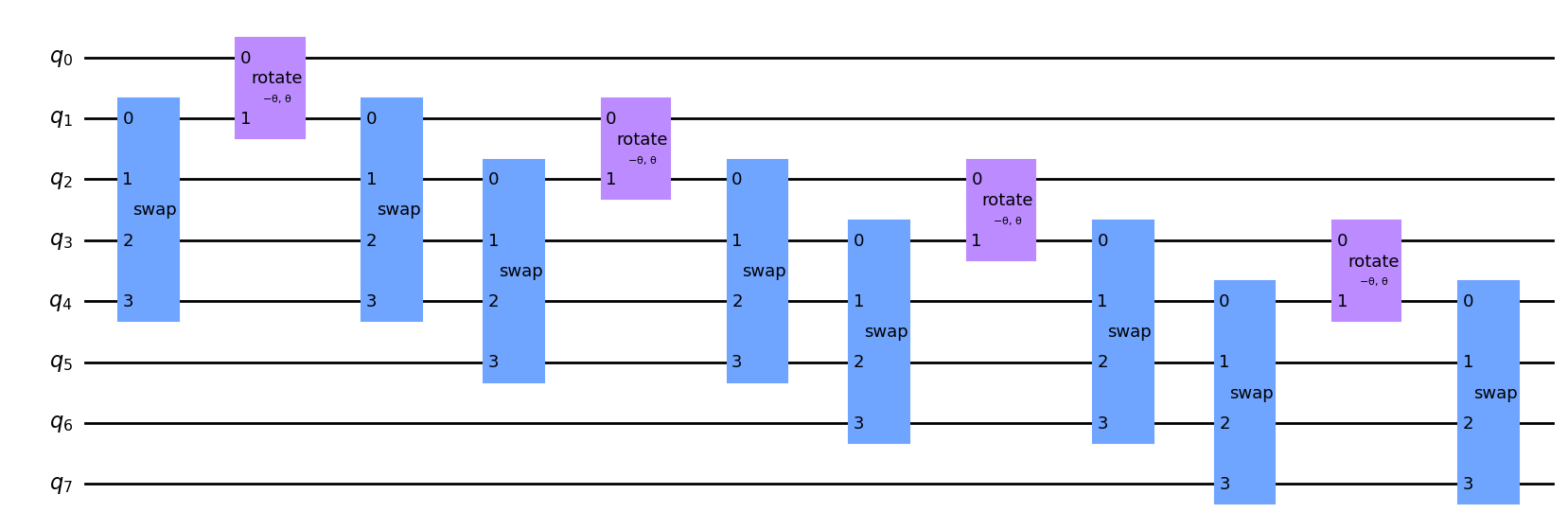}
    \caption{Elementary decomposition of the $4$-qubit convolution unitary. The purple ``rotate'' gadget is $W_\theta^{(1)}$ in Fig~\ref{fig:elementary_decomp}, while the blue ``swap'' gadget is expanded in Fig~\ref{fig:swap_gadget}.}
    \label{fig:total_decomp}
\end{figure}

\subsection{Properties of fermionic convolution}
\label{Appen:prop_con}
In this subsection, we prove several key properties of fermionic convolution; these properties mimic the convolution operations in classical statistics and discrete bosonic systems~\cite{BGJ23a,BGJ23b}. These properties highlight the relation between fermionic quantum convolution and both quantum and classical information theory.
For clarity, we first summarize the properties in the following lemma and leave the central limit theorem to Sec. \ref{app:CLTBound}.

\begin{lemma}[Summary of properties]
  For even states $\rho, \sigma$ and Gaussian unitary $U_G$, fermionic convolution satisfies the following properties:
\begin{enumerate}[(1)]
    \item \emph{Commutativity (Corollary~\ref{cor:com})}: $\rho \boxtimes_{\pi/4} \sigma = \sigma \boxtimes_{\pi/4} \rho$.
    \item \emph{Angle reflectivity (Corollary~\ref{cor:ANG_REF})}: $\rho \boxtimes_\theta \sigma = \rho \boxtimes_{-\theta} \sigma$. 
    \item \emph{Commutativity with Gaussian unitaries (Proposition~\ref{prp:convUnitaryCommute})}: $U_G(\rho \boxtimes \sigma) U^\dag_G = (U_G\rho U^\dag_G) \boxtimes (U_G\sigma U^\dag_G)$, for any Gaussian unitary $U_G$. 
    \item \emph{Convolution-multiplication duality (Theorem~\ref{thm:conDual})}: 
    Let $\Xi_\tau\in \G_{2n}$ denote the Fourier transform of a quantum state $\tau\in \Cl_{2n}$, and 
    $\alpha\cdot \eta$ denote the generators $\alpha \eta_1, \cdots, \alpha \eta_{2n}$, then 
    \begin{eqnarray}
        \Xi_{\rho \boxtimes \sigma}(\eta) = \Xi_\rho(\cos\theta\cdot \eta)\Xi_\sigma(\sin\theta\cdot \eta). 
    \end{eqnarray}
    
    \item \emph{Entropy inequality (Proposition~\ref{prop:entrop_ineq})}:  
        \begin{eqnarray}
                    S(\rho\boxtimes\sigma)\geq \frac{1}{2}S(\rho)+\frac{1}{2}S(\sigma).
        \end{eqnarray}
        Equality holds if and only if $\rho=\sigma$ is equal to a Gaussian state. 
        
    \item \emph{Purity invariance implies normality (Corollary~\ref{cor:purityInvariance})}: 
    $\rho\boxtimes\rho$ is pure if and only if $\rho$ is a Gaussian state. 
\end{enumerate}
\end{lemma}

We begin by proving the convolution-multiplication duality, which will be helpful to prove other results. It was demonstrated for the special case \(\theta = \pi/4\) and \(\rho = \sigma\) in \cite{hudson1980translation}, Proposition 5.1. 
Here, we will provide a generalized proof, inspired by Hudson's approach~\cite{hudson1980translation}, and also a novel, elementary combinatorial proof. We first introduce the following operation on the algebras.

\begin{definition}[Contraction isomorphism]
Given a factor $\alpha\in \complex$, 
    the contraction isomorphism $\xi_\alpha:\G_{2n}\to \G_{2n}$ scales all generators by a scalar, i.e., 
    \begin{eqnarray}
        \xi_\alpha\, \eta_j = \alpha \eta_j,
    \end{eqnarray}
This contraction can be extended to the whole algebra as 
\begin{align}
    \xi_{\alpha}\eta_A=\alpha^{|A|}\eta_A.
\end{align}

\end{definition}

\begin{theorem}[Convolution-multiplication duality]\label{thm:conDual}
    For arbitrary $\rho$ and even state $\sigma$, the moment-generating operator under convolution satisfies the following relation
    \leqalign{eq:convMultDuality}{
        \Xi_{\rho \boxtimes_\theta \sigma} &= (\xi_{\cos\theta}\Xi_\rho)(\xi_{\sin\theta}\Xi_\sigma),
    }
   and the cumulant-generating operator under convolution satisfies the following relation 
   \begin{align}\label{eq:conMut_cu}
\Psi_{\rho\boxtimes_\theta\sigma}=\xi_{\cos\theta}\Psi_{\rho}+\xi_{\sin\theta}\Psi_{\sigma}.
   \end{align}

\end{theorem}
\begin{proof}

Equation~\eqref{eq:convMultDuality} immediately implies~\eqref{eq:conMut_cu} by taking the logarithm. 
To prove~\eqref{eq:convMultDuality}, we only need to prove the following statement and extend by bilinearity: 
\leqalign{eq:easyProofFormula}{
    \Xi_{\gamma_J\boxtimes_{\theta} \gamma_K}(\eta) 
    = (\xi_{\cos\theta} \eta_J)(\xi_{\sin\theta}\eta_K),
}
for all even $\gamma_J, \gamma_K$. The right-hand side of~\ref{eq:easyProofFormula} equals 
\begin{align}
    \cos^{|J|}\theta \sin^{|K|} \theta \prod_{a\in J}  \eta_a \prod_{b\in K} \eta_b,
\end{align}
is $0$ for  $J\cap K\neq\emptyset$, and $  \cos^{|J|}\theta \sin^{|K|} \theta \prod_{a\in J}  \eta_a \prod_{b\in K} \eta_b$ for $J\cap K=\emptyset$ . 
The left-hand side of~\eqref{eq:easyProofFormula} equals 
\leqalign{eq:basisConvolutionExpansion}{
    \gamma_J \boxtimes_{\theta} \gamma_K
    = \tr_2 \left[W_\theta (\gamma_J \ot \gamma_K) W_\theta^\dag\right]
    = \tr_2 \left(W_\theta \gamma_J \gamma_K W_\theta^\dag\right). 
}
Applying equation~\eqref{eq:convUnitaryGeneratorAction} to the operator inside the trace yields 
\leqalign{eq:preTraceExpansion}{
    W_\theta \gamma_J \gamma_K W_\theta^\dag
    &= \prod_{a\in J} (\gamma_j \cos\theta - \gamma_{2n+j}\sin\theta) \cdot \prod_{b\in K} (\gamma_k \sin\theta + \gamma_{2n+k}\cos\theta). 
}
Note that any component containing $\gamma_{2n+l}$ will be annihilated by $\tr_2$. 
If  $J\cap K=\emptyset$, the he only choice of factors in the expansion~\eqref{eq:preTraceExpansion} is
\begin{align}
   \cos^{|J|} \sin^{|J|} \theta \prod_{a\in J} \gamma_a  \prod_{b\in K} \gamma_b \implies 
   \Xi_{\gamma_J\boxtimes_{\theta} \gamma_K}(\eta) = \cos^{|J|} \sin^{|J|} \theta \prod_{a\in J} \eta_a  \prod_{b\in K} \eta_b. 
\end{align}
If $J\cap K\neq\emptyset$, then
\leqalign{eq:dualityLHS}{
    W_\theta \gamma_J \gamma_K W_\theta^\dag 
    &= \pm (\gamma_l \cos\theta - \gamma_{2n+l}\sin\theta)(\gamma_l \sin\theta + \gamma_{2n+l}\cos\theta)(\cdots) \\ 
    &= \pm (\gamma_l^2 \cos\theta \sin\theta - \gamma_{2n+l}^2 \cos\theta\sin\theta + 2(\cos^2\theta + \sin^2\theta) \gamma_l \gamma_{2n+l})(\cdots) \\ 
    &= \pm 2 \gamma_l \gamma_{2n+l}(\cdots) 
}
which is also $0$ after taking $\tr_2$. This proves that the two sides of equation~\eqref{eq:easyProofFormula} are equal, completing the proof. 
\end{proof}

We also provide another proof following Hudson's method in~\cite{hudson1980translation}, which may be of independent interest. To do so, we need an additional isomorphism. 
\begin{definition}[Exchange isomorphism]
       The exchange isomorphism $\mca S:\Cl_{2n}\otimes \Cl_{2n}\to \Cl_{2n}\otimes \Cl_{2n}$ is defined by 
    \begin{eqnarray}
        \mca S (\gamma_j \otimes 1) = 1\otimes \gamma_j, \quad \mca S(1\otimes \gamma_j) = \gamma_j \otimes 1,
    \end{eqnarray}
    extended multiplicatively from generators to the whole algebra.  
\end{definition}

\begin{proof}
    Consider the space $\Cl_{2n}\otimes \Cl_{2n} \otimes_f \G_{2n}$,
    which we relabel as $\Cl_1\otimes \Cl_2 \otimes_f \G$. Here $\Cl_1\otimes \Cl_2$ 
    is the product space of $\rho\otimes \sigma$. Let $e^F\in \Cl_{2n}\otimes_f \G_{2n}$ denote 
    the Fourier kernel. 
    First expand the partial traces, note that $(\mca S\otimes_f 1_{\G\to \G})\, (1_{\Cl_1}\otimes e^F)$ 
    effectively applies $K$ to the $\Cl_1, \G$ components of $\Cl_1\otimes \Cl_2 \otimes_f \G$. 
    \leqalign{eq:startFourierExpression}{
        \Xi_{\rho \boxtimes_{\theta}\sigma}
        &= \tr_{\Cl_1}[e^F (\rho \boxtimes \sigma)] 
        = \tr_{\Cl_1}\left[e^F\, \tr_{\Cl_2}[W(\rho \otimes \sigma)W^\dag]\right] \\ 
        &= \tr_{\Cl_1\Cl_2} \left[ [(\mca S\otimes_f 1_{\G\to \G})\, (1_{\Cl_1}\otimes e^F)]
            (W\otimes_f 1_\G)(\rho \otimes \sigma \otimes_f 1_\G)(W^\dag \otimes_f 1_\G)
        \right] \\ 
        &= \tr_{\Cl_1\Cl_2} \left[(W^\dag \otimes_f 1_\G)
                        [(S\otimes_f 1_{\G\to \G})\, (1_{\Cl_1}\otimes e^F)]
                        (W\otimes_f 1_\G)(\rho \otimes \sigma \otimes_f 1_\G)
        \right].
    }
    Now consider $(W^\dag \otimes_f 1_\G)[(\mca S\otimes_f 1_{\G\to \G})\, (1_{\Cl_1}\otimes K)](W\otimes_f 1_\G)$. 
    Recalling the Fourier kernel expansion~\eqref{eq:FourierKernel}, we have
    \malign{
        (\mca S\otimes_f 1_{\G\to \G})\, (1_{\Cl_1}\otimes e^F) 
        &= \sum_J (\mca S\otimes_f 1_{\G\to \G})\, [(1_{\Cl_1}\otimes 1_{\Cl_2} \otimes_f \eta_J)(1_{\Cl_1}\otimes \gamma_J^\dag \otimes_f 1_\G)] \\ 
        &= \sum_J (1_{\Cl_1}\otimes 1_{\Cl_2} \otimes_f \eta_J)(\gamma_J^\dag \otimes 1_{\Cl_2} \otimes_f 1_\G)\\  
        &= \exp\left[\sum_j (\gamma_j \otimes 1_{\Cl_2} \otimes_f 1_\G)(1_{\Cl_1} \otimes 1_{\Cl_2} \otimes_f \eta_j) \right]. 
    }
    Unitary conjugation by $W^\dag \otimes 1_\G$ acts on the exponent of the expression above as follows
    \malign{
        &(W^\dag \otimes 1_\G)\left[(\gamma_j \otimes 1_{\Cl_2} \otimes_f 1_\G)(1_{\Cl_1} \otimes 1_{\Cl_2} \otimes_f \eta_j)\right] (W\otimes 1_\G)\\
        &= \left[ (W^\dag \otimes 1_\G)(\gamma_j \otimes 1_{\Cl_2} \otimes_f 1_\G)(W\otimes 1_\G) \right] (1_{\Cl_1} \otimes 1_{\Cl_2} \otimes_f \eta_j) \\ 
        &= \left[\cos\theta (\gamma_j \otimes 1_{\Cl_2} \otimes_f 1_\G) + \sin\theta(Z^{\otimes n} \otimes \gamma_j \otimes_f 1_\G)\right] (1_{\Cl_1} \otimes 1_{\Cl_2} \otimes_f \eta_j). 
    }
    This yields conjugate action on $[(\mca S\otimes_f 1_{\G\to \G})\, (1_{\Cl_1}\otimes K)]$
    \leqalign{eq:FourierExponent}{ 
        &(W^\dag \otimes_f 1_\G)[(\mca S\otimes_f 1_{\G\to \G})\, (1_{\Cl_1}\otimes K)](W\otimes_f 1_\G)\\
        =& \exp \left[ \sum_j \left[\cos\theta (\gamma_j \otimes 1_{\Cl_2} \otimes_f 1_\G) + \sin\theta(Z^{\otimes n} \otimes \gamma_j \otimes_f 1_\G)\right] (1_{\Cl_1} \otimes 1_{\Cl_2} \otimes_f \eta_j)\right]. 
    }        
    Note that $Z^{\otimes n}\otimes \gamma_j^{(2n)} \cong \gamma_{2n+j}^{(4n)}$.
    Here, to simplify the notation, we used the Jordan-Wigner representation $Z^{\otimes n}$ to denote its algebraic equivalent. Since $\sigma$ is even, in the power expansion of the exponential, only terms that are even in $\Cl_2$ contribute nontrivially after $\tr_{\Cl_1\Cl_2}$ in equation~\eqref{eq:startFourierExpression}.
    More precisely, let $\Phi: \Cl_1\otimes \Cl_2\otimes_f \G\to \Cl_1\otimes \Cl_2\otimes_f \G$ be the projection operator onto the even subspace of $\Cl_2$ so that $\Phi(\gamma_J\ot \gamma_K \ot_f \eta_L) = \gamma_J\ot \gamma_K \ot_f \eta_L$ iff $\gamma_K\in \Cl_2$ is even else $0$, then
    \leqalign{eq:convolutionIntermediateExp}{
    \Xi_{\rho \boxtimes_{\theta} \sigma} 
    &= \tr_{\Cl_1\Cl_2} \left[(W^\dag \otimes_f 1_\G)
                        [(S\otimes_f 1_{\G\to \G})\, (1_{\Cl_1}\otimes e^F)]
                        (W\otimes_f 1_\G)(\rho \otimes \sigma \otimes_f 1_\G)\right] \\ 
    &= \tr_{\Cl_1\Cl_2} \left[\Phi \left[(W^\dag \otimes_f 1_\G)
                        [(S\otimes_f 1_{\G\to \G})\, (1_{\Cl_1}\otimes e^F)]
                        (W\otimes_f 1_\G)\right](\rho \otimes \sigma \otimes_f 1_\G)
        \right]. 
    }
    The even-parity projection of~\eqref{eq:FourierExponent} in the expression above gets rid of the $Z^{\otimes n}$ factor, yielding 
    \malign{
        \Phi(\ref{eq:FourierExponent}) &= \exp \left[ \sum_j \left[\cos\theta (\gamma_j \otimes 1_{\Cl_2} \otimes_f 1_\G) + \sin\theta(1_{\Cl_1} \otimes \gamma_j \otimes_f 1_\G)\right] (1_{\Cl_1} \otimes 1_{\Cl_2} \otimes_f \eta_j)\right] \\
        &= \exp \left[ \sum_j (\gamma_j \otimes 1_{\Cl_2} \otimes_f 1_\G) (1_{\Cl_1} \otimes 1_{\Cl_2} \otimes_f (\cos\theta \cdot \eta_j))\right]
        \exp \left[ \sum_j (1_{\Cl_1} \otimes \gamma_j \otimes_f 1_\G) (1_{\Cl_1} \otimes 1_{\Cl_2} \otimes_f (\sin\theta \cdot \eta_j))\right] \\ 
        &= (\xi_{\cos\theta}e^{F_1})(\xi_{\sin\theta} e^{F_2}). 
    }
  Here $e^{F_1}, e^{F_2}$ are the Fourier kernels defined on $(\Cl_1, \G)$ and $(\Cl_2, \G)$ and extended trivially,
        \malign{
            e^{F_1} &= \exp \left[\sum_j (\gamma_j \ot 1_{\Cl_2} \ot_f 1_{\G})(1_{\Cl_1} \ot 1_{\Cl_2} \ot_f \eta_j) \right] \in \Cl_1 \ot \Cl_2 \ot_f \G \\ 
            e^{F_1} &= \exp \left[\sum_j (1_{\Cl_1} \ot \gamma_j \ot 1_{\G})(1_{\Cl_1} \ot 1_{\Cl_2} \ot_f \eta_j \right] \in \Cl_1 \ot \Cl_2 \ot_f \G. 
        }
        Then substituting $\Phi(\ref{eq:FourierExponent})= (\xi_{\cos\theta}e^{F_1})(\xi_{\sin\theta} e^{F_2})$ into equation~\eqref{eq:convolutionIntermediateExp} yields 
        \malign{
            \Xi_{\rho \boxtimes \sigma} 
            &= \tr_{\Cl_1\Cl_2} \left[\Phi \left[(W^\dag \otimes_f 1_\G)
                            [(S\otimes_f 1_{\G\to \G})\, (1_{\Cl_1}\otimes e^F)]
                            (W\otimes_f 1_\G)\right](\rho \otimes \sigma \otimes_f 1_\G)
            \right] \\ 
            &= \tr_{\Cl_1\Cl_2} \left[(\xi_{\cos\theta}e^{F_1})(\xi_{\sin\theta}e^{F_2})(\rho \otimes \sigma \otimes_f 1_\G)\right]\\
            &= \tr_{\Cl_1\Cl_2} \left[[(\xi_{\cos\theta}e^{F_1})(\rho \otimes 1_{\Cl_2} \otimes_f 1_\G)][(\xi_{\sin\theta}e^{F_2})(1_{\Cl_1} \otimes \sigma \otimes_f 1_\G)]
            \right] \\ 
            &= \tr_\Cl[(\xi_{\cos\theta}e^F) (\rho\otimes_f 1_\G)]\cdot \tr_\Cl[(\xi_{\sin\theta}K) (\sigma\otimes_f 1_\G)] 
            = (\xi_{\cos\theta}\Xi_\rho)(\xi_{\sin\theta}\Xi_\sigma).
        }
    \end{proof}
According to \eqref{eq:conMut_cu} ,   the components of the cumulants under the convolution satisfies the following relation,
    \leqalign{eq:convCumulantAdditivity}{
        \kappa^{\rho \boxtimes_\theta \sigma}_J = \kappa^\rho_J \cos^j \theta + \kappa^\sigma_J \sin^j \theta. 
    }
That is, the quadratic cumulants are preserved, while the higher-order cumulants shrink.
    
\begin{corollary}[Fixed points of convolution]
    \label{cor:gaussianFixedPoint}
    For any Gaussian state $\rho_G$, 
    $\rho_G\boxtimes_{\theta}\rho_G=\rho_G$.
\end{corollary}

\begin{corollary}[Angle-reflectivity]\label{cor:ANG_REF}
    For any even states $\rho$ and $\sigma$, we have 
    $\rho \boxtimes_\theta \sigma = \rho \boxtimes_{-\theta} \sigma$ for all $\theta$. 
\end{corollary}
\begin{proof}
    Note that any even state is determined by the cumulant. It is easy to see that the equation~\ref{eq:convCumulantAdditivity} is invariant under 
    $\theta\mapsto -\theta$ for even $\rho, \sigma$, hence $\rho \boxtimes_\theta \sigma = \rho \boxtimes_{-\theta} \sigma$ .

\end{proof}

\begin{corollary}[Commutativity]\label{cor:com}
 For any even states $\rho$ and $\sigma$, we have $\rho \boxtimes_\theta \sigma = \sigma \boxtimes_\theta \rho$ when $\cos\theta = \sin\theta$. 
\begin{proof}
  Based on the equation~\eqref{eq:convCumulantAdditivity}, 
  the invariance is also preserved under reflection of $\rho$ and $\sigma$ when $\cos\theta = \sin\theta$. 
\end{proof}
\end{corollary}

\begin{proposition}[Commutativity with Gaussian unitaries]
    \label{prp:convUnitaryCommute}
    For any Gaussian unitary $U_G$, we have
    \begin{eqnarray}
        U_G(\rho \boxtimes \sigma)U^\dag_G = (U_G\rho U^\dag_G)\boxtimes (U_G\sigma U^\dag_G).
    \end{eqnarray}

\end{proposition}
    \begin{proof}
   We only need to prove that
    $\Xi_{U_G(\rho \boxtimes \sigma)U^\dag_G}=\Xi_{(U_G\rho U^\dag_G)\boxtimes (U_G\sigma U^\dag_G)}$. 
    Since the conjugate action of the Gaussian unitary $U_G$ corresponds to a rotation $R$ as in \eqref{eq:gaussianUnitaryGrassmannAction}, then 
    left-hand side can be written as
    \[ 
        \Xi_{U_G(\rho \boxtimes \sigma)U^\dag_G}(\eta) 
        = \Xi_{\rho \boxtimes \sigma}(R\eta) 
        = \Xi_\rho(R\eta)\Xi_\sigma(R\eta). 
    \]
    Apply convolution-multiplication duality,  the right-hand side can be 
     \[ 
     \Xi_\rho(R\eta)\Xi_\sigma(R\eta) = \Xi_{U_G\rho U^\dag_G}(\eta)\Xi_{U_G\sigma U^\dag_G}(\eta)= \Xi_{(U_G\rho U^\dag_G)\boxtimes (U_G\sigma U^\dag_G)}(\eta).
     \] 
     The Fourier transforms are equal iff the states are equal, completing the proof. 
    \end{proof}
    
\begin{lemma}[Angle-reflection$+$exchange symmetry of the convolution unitary]
    For all even basis elements $\gamma_J, \gamma_K$, 
    \leqalign{eq:exchangeBehavior}{
        \mca S[W_\theta(\gamma_J \otimes \gamma_K) W_\theta^\dag] = W_{-\theta} (\gamma_K \otimes \gamma_J) W_{-\theta}^\dag.
    }

\end{lemma}
    \begin{proof}
        Let $\varphi_\theta:\Cl_{2n} \otimes \Cl_{2n}\to \Cl_{2n}\otimes \Cl_{2n}$ denote the isomorphism $\varphi_\theta(A) = W_\theta A W_\theta^\dag$, then~\eqref{eq:exchangeBehavior} claims $\mca S\circ \varphi_\theta = \varphi_{-\theta}$ on the even subspace of $\Cl_{2n}\otimes \Cl_{2n}$. 
        To see this, recall the conjugate action of the convolution unitary,
        \malign{
            \varphi_\theta(\gamma_j)=\cos\theta \gamma_j-\sin\theta \gamma_{2n+j}, \\
            \varphi_\theta(\gamma_{2n+j})=\sin\theta \gamma_j+\cos\theta \gamma_{2n+j},
        }
        is invariant under $\theta \mapsto -\theta, \gamma_j\leftrightarrow \gamma_{2n+j}$, where $\gamma_j\leftrightarrow \gamma_{2n+j}$ is exactly the action of $\mca S$ on $\Cl_{2n}\otimes \Cl_{2n}$ 
        \emph{for even elements}. 
    \end{proof}

\begin{lemma}\label{lem:comple}
    For any even states $\rho, \sigma$, let $\rho \,{\tilde \boxtimes_\theta}\, \sigma := \tr_1 \left[W_\theta(\rho \ot \sigma) W_\theta^\dag\right] $ be the complementary channel, 
    we have 
    \begin{align}
        \rho \,{\tilde \boxtimes_\theta}\, \sigma = \rho \boxtimes_{\theta} \sigma.
    \end{align}

\end{lemma}

    \begin{proof}
        $
            \rho \,{\tilde \boxtimes_{-\theta}}\, \sigma 
            = \tr_1 \left[W_\theta(\rho \ot \sigma) W_\theta^\dag\right] 
            = \tr_1 \left[\varphi_\theta (\rho \ot \sigma)\right] 
            = \tr_2 \left[ (\mca S \circ \varphi_\theta)(\rho \ot \sigma) \right] 
            = \tr_2 \left[ \varphi_{-\theta}(\rho \ot \sigma)\right] = \rho \boxtimes_{-\theta} \sigma 
        $. 
    \end{proof}

\begin{proposition}[Entropy inequality under fermionic convolution]
\label{prop:entrop_ineq}
    Given two even states $\rho$ and $\sigma$, we have
    \begin{eqnarray}
                S(\rho\boxtimes_{\theta}\sigma)\geq \frac{1}{2}S(\rho)+\frac{1}{2}S(\sigma),
    \end{eqnarray}
    equality holds iff $\rho=\sigma$ is a Gaussian state. 
    \begin{proof}
        Recall that $\rho \boxtimes_\theta \sigma$, $\rho\,{\tilde \boxtimes}_\theta\, \sigma$ are the 
        two reduced states of $\tau_{ab}=W(\rho \ot \sigma) W^\dag$, then 
        by the subadditivity of quantum entropy and Lemma~\ref{lem:comple},
        \begin{eqnarray}
         S(\rho)+S(\sigma) =  S(\tau_{AB})\leq S(\rho\boxtimes_{\theta}\sigma)+S(\rho\, \tilde{\boxtimes}_{\theta}\, \sigma)
            =2S(\rho\boxtimes_{\theta}\sigma). 
        \end{eqnarray}
        Equality holds iff subadditivity is saturated, which means
        \malign{
            W (\rho \ot \sigma) W^\dag
            =(\rho\boxtimes_{\theta}\sigma)\ot (\rho\,\tilde{\boxtimes}_{\theta}\,\sigma).
        }
    This is also equivalent to
    \begin{align}
         \rho\ot\sigma 
            = W^\dag [(\rho\boxtimes_{\theta}\sigma)\ot (\rho\boxtimes_{\theta}\sigma)] W. 
    \end{align}
       By taking the two reduced states of the right hand side, which are equal to  $\rho$ and $\sigma$ respectively. Thus,
        \begin{eqnarray}
            \rho=(\rho\boxtimes_{\theta}\sigma)\boxtimes (\rho\boxtimes_{\theta}\sigma)
            =(\rho\boxtimes_{\theta}\sigma)\, \tilde{\boxtimes}_{\theta}\, (\rho\boxtimes_{\theta}\sigma)
            =\sigma.
        \end{eqnarray}
       Moreover, by applying the previous argument iteratively and the central limit theorem for fermionic convolution in Proposition~\ref{cor:clt}, we have 
        \begin{eqnarray}
            \rho=(\rho\boxtimes_{\theta}\rho)\boxtimes_{\theta} (\rho\boxtimes_{\theta}\rho)
            =...=\mca G(\rho).
        \end{eqnarray}
       Therefore, $\rho=\sigma$ is a fermionic Gaussian state.
    \end{proof}
\end{proposition}
One direct implication of the entropy inequality is the monotonicity of quantum entropy under fermionic convolution.
\begin{proposition}\label{prop:mon_en}
For any even state $\rho$, we have 
\begin{align}
    S(\boxtimes^{k+1}\rho)\geq S(\boxtimes^k\rho), \forall k.
\end{align}
\end{proposition}

\begin{corollary}[Purity invariance implies normality] 
    \label{cor:purityInvariance}
    $\rho \boxtimes \rho$ is pure if and only if $\rho$ is a pure Gaussian state. 
\end{corollary}
\begin{proof}
If $\rho$ is pure Gaussian, then $\rho \boxtimes \rho = \rho$ is pure by Corollary~\ref{cor:gaussianFixedPoint}. 
The converse implies $S(\rho)=0$ and the saturation condition of the entropy inequality $\rho = \mca G(\rho)$. 
\end{proof}

\begin{proposition}[Gaussian characterization by maximum entropy]
Let $M$ be a real and antisymmetric $2n\times 2n$ matrix with $M^TM \leq I$, and
$\mca Q(M)$ be the set of all even $n$-qubit quantum 
    states with covariance matrix $M$, then 
    the unique Gaussian state $\rho_G\in \mca Q(M)$ has maximum entropy. 

    \begin{proof}
        The Gaussian $\rho_G\in \mca Q(M)$ is unique since Gaussian states 
        are specified by their covariance matrix. Moreover, since 
        covariance matrix is preserved under convolution, i.e., 
        $\Sigma_{\rho\boxtimes\rho}=\Sigma_{\rho}$, then $\rho\boxtimes \rho$ is also in 
        $\mca Q(M)$ for any $\rho\in \mca Q(M)$.
        Let $\rho_0=\mrm{argmax}_{\rho \in \mca Q(M)} S(\rho)$, then 
        $S(\rho_0\boxtimes \rho_0)= S(\rho_0)$. By the entropy 
        inequality in Proposition~\ref{prop:entrop_ineq}, we have  that $\rho_0$ is a Gaussian state. 
    \end{proof}
\end{proposition}
\begin{proposition}[Gaussian characterization by maximum Gaussian weight]
    \label{prp:gaussianMaxQuadraticWeight}
    Given $0\leq S_0 \leq n\log 2$, define 
    \begin{eqnarray}
        \mca R(S_0) = \{n\text{-qubit state } \rho : \rho \text{ even and } S(\rho) = S_0\}. 
    \end{eqnarray}
    If $\rho\in \mca R(S_0)$ has maximum Gaussian weight $K_G$, then is Gaussian. 
    \begin{proof}
        For the sake of contradiction, suppose that a non-Gaussian state $\rho \in \mca R(S_0)$ has maximum Gaussian weight $K_G$, then $S(\mca G(\rho)) > S(\rho)$. Then applying Lemma~\ref{lem:entropyQuadraticweightInequality} to $\mca G(\rho)$ with $S_0=S(\rho)$ implies the existence of a Gaussian state $\sigma$ satisfying $S(\sigma) = S_0 < S(\rho)$ and $K_G(\sigma) > K_G(\rho)$.
    \end{proof}
\end{proposition}

Applying this to the extremal case $S_0=0$, we obtain the Gaussian-weight characterization of pure Gaussian states. 
\begin{corollary} A $n$-qubit even state $\psi$ is Gaussian if and only 
if $K_G(\psi)= n$. 
\end{corollary}

\subsection{Properties of non-Gaussian entropy and its generalization to R\'enyi entropy}
\label{app:renyiGeneralization}
\begin{figure}[t]
  \centering
  \begin{minipage}{0.48\textwidth}
    \centering
    \includegraphics[width=1\linewidth]{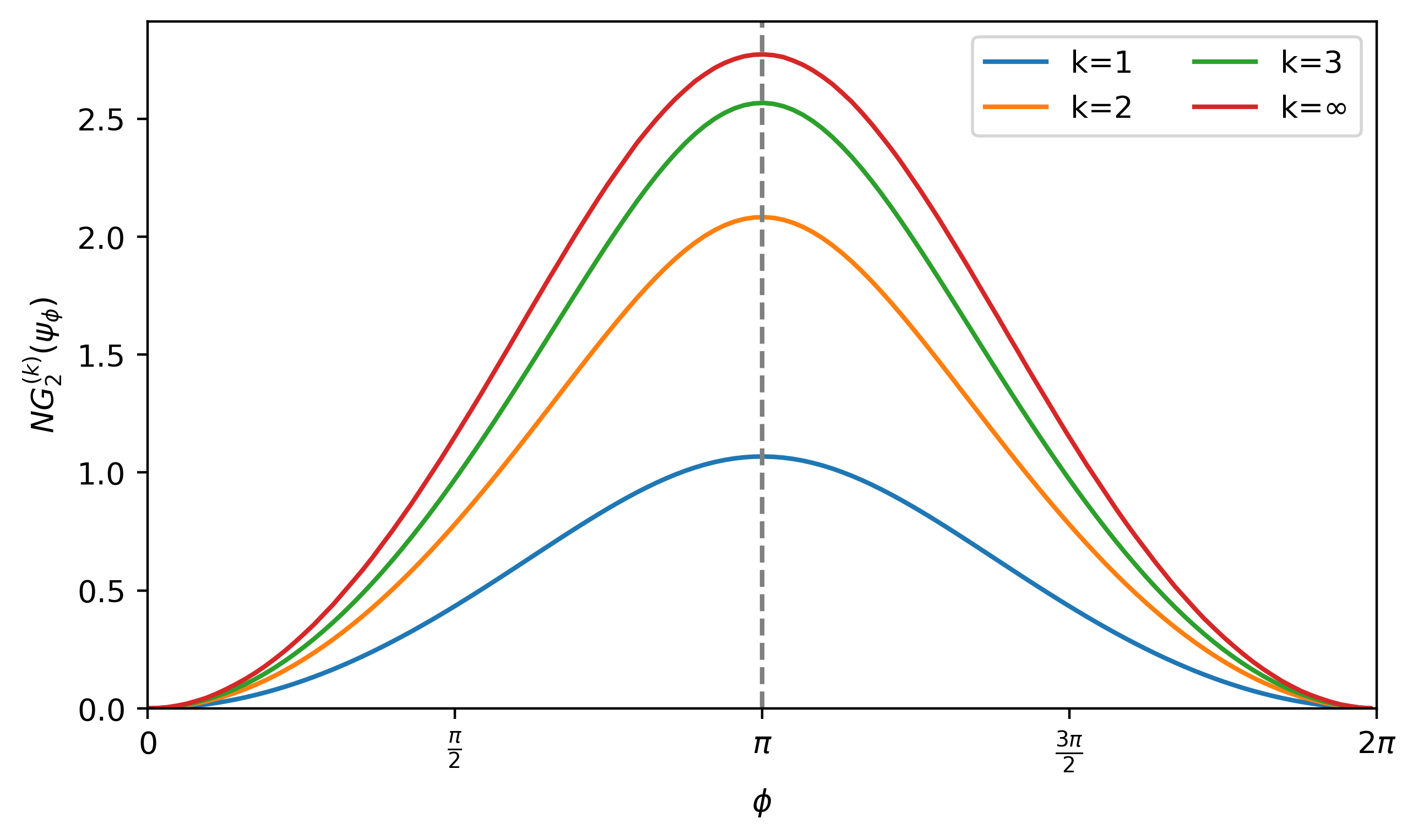}
    \caption{
        Non-Gaussian R\'enyi entropy $NG^{(k)}_2$  of the 4-qubit parametrized state $\psi_{\phi}$ in ~\eqref{eq:4qubitFamily}. 
    }
    \label{fig:renyi2Plot}
  \end{minipage} 
  \begin{minipage}{0.48\textwidth}
    \centering
    \includegraphics[width=1\linewidth]{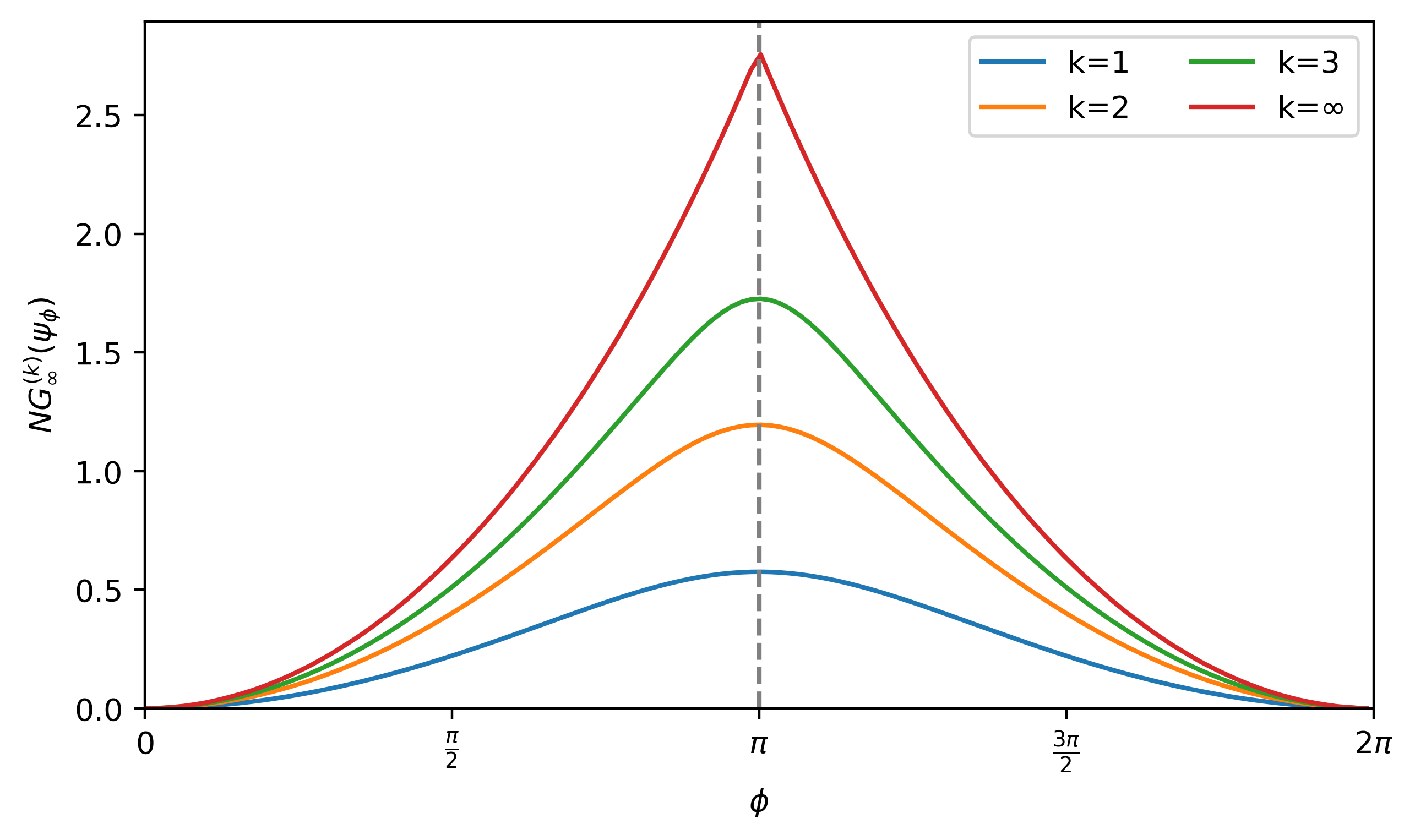}
    \caption{Non-Gaussian R\'enyi entropy $NG^{(k)}_{\infty}$of the 4-qubit parametrized state $\psi_{\phi}$ in ~\eqref{eq:4qubitFamily}. }
    \label{fig:renyiInfPlot}
  \end{minipage}
\end{figure}

Recall that the non-Gaussian entropy is defined as 
\begin{align}
    NG(\psi)=S(\psi\boxtimes\psi).
\end{align}
\begin{theorem}[Properties of non-Gaussian entropy]
     The non-Gaussian entropy $NG(\psi)$ for an even pure state $\psi$  satisfies the following properties:
   \begin{enumerate}[(1)]
    \item Faithfulness: $NG(\psi)\geq 0$ with equality iff $\psi$ is a fermionic Gaussian. 
    \item Gaussian invariance: $NG(U_G\proj{\psi} U^\dag_G)=NG(\proj{\psi})$ for any Gaussian unitary $U_G$. 
    \item Additivity under tensor product: $NG(\psi_1\ot\psi_2)=NG(\psi_1)+NG(\psi_2)$. 
\end{enumerate}
\end{theorem}
\begin{proof}

(1) $NG(\psi)\geq 0$ follows from definition with equality iff $S(\psi\boxtimes\psi)=S(\psi)$. 
By the saturation condition of the entropy inequality in Proposition~\ref{prop:entrop_ineq}, we have $NG(\psi)=0$ iff $\psi$ is a fermionic Gaussian.

(2) Gaussian unitary-invariance follows from convolution commuting with Gaussian unitary in Proposition~\ref{prp:convUnitaryCommute}. 

(3) $NG$-additivity follows from 
$(\psi_1\ot\psi_2)\boxtimes (\psi_1\ot\psi_2)= (\psi_1\boxtimes\psi_2)\otimes(\psi_1\boxtimes\psi_2)$. 

\end{proof}

Consider extending non-Gaussian entropy to higher-order, R\'enyi non-Gaussian entropy as follows
\begin{align}
    NG^{(k)}_{\alpha}(\psi)=S_{\alpha}(\boxtimes^k\psi),
\end{align}
where the $\alpha$-R\'enyi entropy $S_{\alpha}$  with $\alpha\in [0, +\infty]$ is defined as 
\begin{eqnarray}
S_{\alpha}(\rho)=\frac{1}{1-\alpha}\Tr{\rho^\alpha}.
\end{eqnarray}
The R\'enyi generalization $NG^{(k)}_{\alpha}$ also satisfies the
essential properties of faithfulness, Gaussian invariance, and additivity under tensor product, making it a valuable measure for quantifying fermionic non-Gaussianity.

Let us again consider the $4$-qubit parametrized state $\psi_\phi$ in ~\eqref{eq:4qubitFamily}, 
 $ \ket{\psi_{\phi}}=\frac{1}{2}
    \left(\ket{0000}+\ket{0011}+\ket{1100}+e^{i\phi}\ket{1111} \right)$,  
with $\phi\in [0,2\pi]$. This degenerates to the von-Neumann definition (Fig~\ref{fig:NG}) with $\alpha=1$. The $\alpha=2, \infty$ cases are shown in Fig~\ref{fig:renyi2Plot} and~\ref{fig:renyiInfPlot}.  
When $\alpha=0$, for all $k>0$ we have 
\begin{eqnarray}
    NG^{(k)}_0(\psi_\phi) = \begin{cases}
        0, & \phi \in \{0, 2\pi\}, \\ 
        4\log 2, & \text{otherwise.}
    \end{cases}
\end{eqnarray}

\section{Central limit theorem and convergence rate}
\label{app:CLTBound}
This section derives the central limit for fermionic convolution and provides a bound on its rate of convergence. 
Let us define the $k$-th iteration of self-convolution 
\begin{align}
    \boxtimes^{k+1}_{\theta}\rho = (\boxtimes^k \rho)\boxtimes_{\theta} (\boxtimes^k \rho),\quad \text{with}~\quad \boxtimes^0\rho = \rho.
\end{align}

\begin{proposition}[Central limit theorem]\label{cor:clt}
For any even state $\rho$, we have
    \begin{eqnarray}
        \lim_{k\to \infty} \boxtimes^k_{\theta} \rho = \mca G(\rho).
    \end{eqnarray}
  
\end{proposition}
  \begin{proof}
        We prove the results for $\theta=\pi/4$, general case follows easily. 
        An even state is Gaussian if and only if it has vanishing super-quadratic cumulants, i.e., $\kappa_J=0$ for any $J$ with $|J|\geq 4$. 
  The convolution-multiplication duality~\eqref{eq:conMut_cu} implies 
  \begin{align*}
        \kappa^{\rho \boxtimes_\theta \rho}_J = \kappa^\rho_J \cos^j \theta + \kappa^\rho_J \sin^j \theta. 
  \end{align*}
    This yields the following equation for $j$-th order cumulants 
        \leqalign{eq:cumulantIter}{
            \kappa^{\boxtimes^{k+1}\rho}_J
            &= 2 \left(\df 1 {\sqrt 2}\right)^{j} \kappa^{\boxtimes^k \rho}_J 
            = 2^{k(1-j/2)} \kappa^{\rho}_J. 
        }
        Note that $j=|J|$ is always even. Therefore, 
        \[ 
            \lim_{k\to \infty} 2^{k(1-j/2)} \kappa^{\rho}_J = \begin{cases}
                \kappa^\rho_J & j=2, \\ 
                0 & \text{otherwise}. 
            \end{cases}
        \]  
        These are exactly the cumulants of $\mca G(\rho)$. 
    \end{proof}

Next, we study the convergence rate of the central limit theorem (CLT). We prove results for $\theta = \pi/4$, and the approach can be readily extended to more general cases.
    
\begin{lemma}
    \label{lem:exponentialBound}
For any linear operators $X, Y$, we have 
   \[ 
        \|e^{X+Y} - e^X\| \leq \|Y\| e^{\|X\|+\|Y\|},
    \]
where $\|A\| = \sup_{\|\vec x\| \leq 1} \|A\vec x\|$ denotes the operator norm.

\end{lemma}
\begin{proof}
        Denote $x=\|X\|, y=\|Y\|$. 
        Expand the series and use submultiplicativity and subadditivity: 
        \malign{
            \|e^{X+Y} - e^X\| 
            &\leq \sum_{n=0}^\infty \df 1 {n!} \|(X+Y)^n - X^n\| \\
            &\leq \sum_{n=0}^\infty \df 1 {n!} \sum_{k=0}^\infty \binom{n}{k} x^k y^{n-k} - x^n \\ 
            &= \sum_{n=0}^\infty \df 1 {n!} \sum_{k=0}^{n-1} \binom n k x^k y^{n-k} \\
            &= \sum_{n=0}^\infty \df 1 {n!} \sum_{k=0}^{n-1} \binom{n-1}{k-1} x^k y^{n-k} + 
            \sum_{k=0}^{n-1} \binom{n-1}{k} x^k y^{n-k}.
        }
        Let us consider the first term in the sum, substitute $k\mapsto k+1$ and factor out an $x$ 
        \[ 
            \sum_{k=0}^{n-1} \binom{n-1}{k-1} x^k y^{n-k} 
            = x\sum_{k=0}^{n-2} \binom{n-1}{k} x^k y^{(n-1)-k} = x[(x+y)^{n-1} - x^{n-1}] = x(x+y)^{n-1} - x^n .
        \]
        The second term is defined similarly 
        \[ 
            \sum_{k=0}^{n-1} \binom{n-1}{k} x^k y^{n-k} = y(x+y)^{n-1}.
        \]
        Substitute back into the expression, which now simplifies neatly 
        \[
            \|e^{X+Y} - e^X\| \leq \sum_{n=0}^\infty \df{(x+y)^n - x^n}{n!} 
            = e^{x+y} - e^x = e^x(e^y-1) \leq y e^xe^y .
        \]
    \end{proof}
    
\begin{lemma}
    \label{lem:l2OperatorBound}
    For $n$-qubit states $\rho, \sigma$, we have
    \[ 
        \|\rho - \sigma\|_2 
        = \df 1 {2^n} \|\Xi_\rho - \Xi_\sigma\|_2 
        \leq \|\Xi_\rho - \Xi_\sigma\|,
    \]
    where $\|\cdot \|_2$ denotes the $L_2$-norm,

\end{lemma}
    \begin{proof}
    This equation comes from the fact that $L_2$ norm is preserved under the Grassmann-Clifford Fourier transform. The inequality comes from the bound on the operator norm by the $L_2$ norm. 
    \end{proof}

\begin{theorem}[Convergence rate of  central limit theorem]
\label{app:convergenceBound} 
For a $n$-qubit even state $\rho$ 
\begin{eqnarray}
    \|\boxtimes^k \rho - \mca G(\rho)\|_2 &\leq \df {\sqrt{K_M(\rho)}} {2^k}
    \exp\left(\sqrt{K_G(\rho)} + 2^{-k} \sqrt{K_M(\rho)}\right),
\end{eqnarray}
where $K_G$ and $K_M$ denote the Gaussian and non-Gaussian weights in definition~\ref{def:GaussianWeights}. 

\begin{proof}
    Recall the cumulant-generating operator under the convolution is equal to
    \[
        \Psi_{\boxtimes^k\rho } = \sum_J 2^{-k(j/2-1)} \kappa^\rho_J \eta_J,
    \]
    while the limit state $\mca G(\rho)$ has the following form
       \[
          \Psi_{\mca G(\rho)} = \sum_{|J|=2} \kappa^\rho_J \eta_J .
    \]
    The operator norm of $\Psi_{\boxtimes^k\rho } - \Psi_{\mca G(\rho)}$ is bounded by the 
    $L_2$ norm 
    \[
        \|\Psi_{\boxtimes^ k\rho } - \Psi_{\mca G(\rho)}\| 
        \leq \|\Psi_{\boxtimes^ k\rho } - \Psi_{\mca G(\rho)}\|_2
        = \sqrt{\sum_{|J|>2} |2^{-k(|J|/2-1)}\kappa_J|^2} \leq 2^{-k} \sqrt{K_M(\rho)}. 
    \]
    By using Lemma~\ref{lem:l2OperatorBound} and substituting $X\mapsto \Psi_{\mca G(\rho)}, 
        Y\mapsto \Psi_{\boxtimes^k \rho} - \Psi_{\mca G(\rho)}$: 
    \[ 
        \|\boxtimes^k \rho - \mca G(\rho)\|_2 
        \leq \|\Xi_{\boxtimes^k \rho} - \Xi_{\mca G(\rho)}\| 
        = \|\exp \Psi_{\boxtimes^k \rho} - \exp \Psi_{\mca G(\rho)}\| 
        = \| e^{X+Y} - e^X\|. 
    \] 
    Invoke Lemma~\ref{lem:exponentialBound} to obtain the desired bound: 
    \[ 
        \|\boxtimes^k \rho - \mca G(\rho)\|_2 
        \leq \| e^{X+Y} - e^X\| 
        \leq ye^{x+y}
        = \df {\sqrt{K_M(\rho)}} {2^k}
        \exp\left(\sqrt{K_G(\rho)} + 2^{-k} \sqrt{K_M(\rho)}\right). 
    \]
    where the relevant quantities $y=\|Y\|, x=\|X\|$ are 
    \[ 
        x = \|\Psi_{\mca G(\rho)}\| \leq \|\Psi_{\mca G(\rho)}\|_2 = \sqrt{K_G(\rho)}, \quad y = \|\Psi_{\boxtimes^k \rho} - \Psi_{\mca G(\rho)}\| \leq 2^{-k}\sqrt{K_M(\rho)}. 
    \]
\end{proof}
\end{theorem}
One may also consider a version of the central limit theorem which consumes an arbitrary positive 
integer $n$ instead of $2^k$ number of copies of the state $\rho$. 
For each $m>0$, define the convolution angle $\theta_m\in (0, \pi/2)$ which satisfies 
\leqalign{def:convAngle}{
    \cos^2 \theta_m = \df{m}{m+1}.
}
Consider a version of iterated convolution linear in the number of copies so that $\boxtimes^{k}\rho$ requires $k$ copies of $\rho$
\leqalign{def:linearIterConv}{
   \boxtimes^{(1)} \rho = \rho, \quad \boxtimes^{(k+1)} \rho = (\boxtimes^{(k)}\rho) \boxtimes_{\theta_k} \rho.
}
In particular, 
$
    \boxtimes^{(2)} \rho = \rho \boxtimes_{\pi/4}\rho, \,\,
    \boxtimes^{(3)} \rho = (\rho \boxtimes_{\pi/4}\rho)\boxtimes_{\theta_2} \rho
$, etc. 
We now compute the cumulants of iterated convolution by applying additivity~\eqref{eq:convCumulantAdditivity}. Let us first consider the quadratic case 
\leqalign{eq:secOrderCumulants}{
    \kappa^{\boxtimes^{(n)}\rho}_J 
    &= \kappa^\rho_J \sin^2 \theta_n + \kappa^{\boxtimes^{(n-1)}\rho}_J \cos^2 \theta_n 
    = \frac 1 n \kappa^\rho_J + \df{n-1}{n} \kappa^{\boxtimes^{(n-1)}\rho}_J \\ 
    &= \frac 1 n \kappa^\rho_J + \df{n-1} n \left(\frac 1 {n-1} \kappa^\rho_J + \frac{n-1}{n-2} \kappa^{\boxtimes^{(n-2)}\rho}_J\right)\\ 
    &= \cdots = \kappa^\rho_J , \quad \text{where}~|J|=2. 
}
In general, for the cumulant with  $|J|=2m$, we have 
\leqalign{eq:generalAvgCumulants}{
    \kappa^{\boxtimes^{(n)}\rho}_J 
    &= \kappa^\rho_J \sin^{2m} \theta_n + \kappa^{\boxtimes^{(n-1)}\rho}_J \cos^{2m} \theta_n 
    = \frac 1 {n^m} \kappa^\rho_J + \left(\df{n-1}{n}\right)^m \kappa^{\boxtimes^{(n-1)}\rho}_J  \\ 
    &= \frac 1 {n^m} \kappa^\rho_J + \left(\df{n-1}{n}\right)^m \left[
        \frac 1 {(n-1)^m} \kappa^\rho_J + \left(\df{n-2}{n-1}\right)^m \kappa^{\boxtimes^{(n-2)}\rho}_J 
    \right] \\ 
    &= \frac 1 {n^m} \kappa^\rho_J + \df 1 {n^m} \kappa^\rho_J + \left(\df{n-2} n\right)^m \kappa^{\boxtimes^{(n-2)}\rho}_J  \\ 
    &= \frac 1 {n^m} \kappa^\rho_J + \df 1 {n^m} \kappa^\rho_J + \left(\df{n-2} n\right)^m \left[
       \frac 1 {(n-2)^m} \kappa^\rho_J + \left(\df{n-3}{n-2}\right)^m \kappa^{\boxtimes^{(n-3)}\rho}_J
    \right] \\ 
    &= \frac 1 {n^m} \kappa^\rho_J + \df 1 {n^m} \kappa^\rho_J + \df 1 {n^m} \kappa^\rho_J + \left(\df{n-3}{n}\right)^m \kappa^{\boxtimes^{(n-3)}\rho}_J \\ 
    &= \cdots = \df{1}{n^{m-1}}\kappa^\rho_J. 
}
Applying this to the proof in theorem~\ref{app:convergenceBound}, we obtain  
\begin{eqnarray}
    \|\Psi_{\boxtimes^{(n)}\rho } - \Psi_{\mca G(\rho)}\| 
    \leq \|\Psi_{\boxtimes^{(m)}\rho } - \Psi_{\mca G(\rho)}\|_2
    = \sqrt{\sum_{|J|>2} n^{1-|J|/2}\kappa_J|^2} \leq \df 1 n \sqrt{K_M(\rho)}. 
\end{eqnarray}
Following the line of reasoning in the proof, the final answer replaces $\sqrt{K_M(\rho)}/2^k \mapsto \sqrt{K_M(\rho)}/n$ and results in 
\leqalign{eq:anyCopyCLTBound}{
    \|\boxtimes^{(n)} \rho - \mca G(\rho)\|_2 &\leq \df {\sqrt{K_M(\rho)}} {n}
        \exp\left(\sqrt{K_G(\rho)}+ \sqrt{K_M(\rho)}/n\right). 
}

\section{Properties of Gaussification }\label{appen:Gassification}

\begin{lemma}
    For any fermionic state $\rho$ and fermionic Gaussian unitary $U$, 
    \begin{align}
        \mca G(U\rho U^\dag) =U\mca G(\rho)U^\dag.
    \end{align}
    \begin{proof}
      Apply the central limit theorem in Proposition~\ref{cor:clt} and convolution commuting with Gaussian unitaries in Proposition~\ref{prp:convUnitaryCommute} to obtain 
        \[ 
            \mca G(U\rho U^\dag) = \lim_{k\to \infty} \boxtimes^k (U\rho U^\dag)
            = U\left(\lim_{k\to \infty} \boxtimes^k \rho\right) U^\dag = U\mca G(\rho) U^\dag. 
        \] 
    \end{proof}
\end{lemma}

\begin{lemma}\label{lem:comm_GL}
    For even $\rho$ and $\sigma$ with $\mca G(\rho)=\mca G(\sigma)$, 
    $\mca G(\rho \boxtimes \sigma) = \mca G(\rho)$. 
    \begin{proof}
        Convolution preserves quadratic cumulants, and $\mca G(\tau_1) = \mca G(\tau_2)$ iff their quadratic cumulants are equal. 
    \end{proof}
\end{lemma}

\begin{theorem}
For any even state $\rho$, the following equation holds 
\begin{align}
\min_{\sigma_G\in \textbf{Gauss} }D(\rho||\sigma_G)
=D(\rho||\mca G(\rho))
=S(\mca G(\rho))-S(\rho),
\end{align}
where the set $\textbf{Gauss}$ denotes the set of all fermionic Gaussian states. 
\end{theorem}
\begin{proof}
    For every state $\rho$, $\mca G(\rho)$ has the same covariance  as $\rho$. Then 
    \begin{align*}
       & D(\rho||\sigma_G)-D(\rho||\mca G(\rho))\\
        =&\Tr{\rho\log\rho}-\Tr{\rho\log\sigma_G}
        -\Tr{\rho\log\rho}+\Tr{\rho\log\mca G(\rho)}\\
        =&\Tr{\rho\log\mca G(\rho)}-\Tr{\rho\log\sigma_G}\\
        =&\Tr{\mca G(\rho)\log\mca G(\rho)}-\Tr{\mca G(\rho)\log\sigma_G}\\
        =&D(\mca G(\rho)||\sigma_G),
    \end{align*}
    where the second to the last equality comes from the fact that states
    $\rho$ and $\mca G(\rho)$ have the same covariance, and both $\mca G(\rho), \sigma_G$ have the  quadratic structure $\exp(\gamma^T h\gamma)$. 
\end{proof}

We obtain the following corollary based on this theorem and $\mca G(\boxtimes^k\rho)=\mca G(\rho)$ from Lemma~\ref{lem:comm_GL}: 
\begin{corollary}
       For any fermionic state $\rho$, we have
\begin{align}
\min_{\sigma_G\in \textbf{Gauss} }D(\boxtimes^k\rho||\sigma_G)
=D(\boxtimes^k\rho||\mca G(\rho))
=S(\mca G(\rho))-S(\boxtimes^k\rho),
\end{align}
for any integer $k$, 
where $\textbf{Gauss}$ denotes the set of all fermionic Gaussian states.

\end{corollary}

    

\section{Gaussian unitary testing}
\label{app:gaussianUnitaryTest}
In this section, we use the Choi–Jamiołkowski isomorphism to extend convolution results from states to unitary channels. The Choi–Jamiołkowski isomorphism in the 
context of the fermion algebra has first been considered in~\cite{bravyi2004lagrangian}, 
in which it was shown that the Fourier representation of completely-positive 
trace-preserving Gaussian channels have a Lagrangian integral representation 
in the Grassmann algebra. Here, we characterize the conditions for the Choi 
state of a unitary channel to be fermionic Gaussian. 
By doing this, we are able to prove that a unitary channel is 
Gaussian iff it is even and has a fermionic Gaussian Choi state. 

The fermionic maximally-entangled state $\rho_I\in \Cl_{4n}$ on $n+n$ qubits is (\cite{bravyi2004lagrangian}, Definition $6$)
\begin{eqnarray}
    \rho_I = \df 1 {2^{2n}} \prod_{j=1}^{2n} (1 + i\gamma_j \gamma_{2n+j}). 
\end{eqnarray}
Note that this is not the usual maximally-entangled state $\sum_{x\in \{0, 1\}^n} |x\ra \ot |x\ra$, which is not fermionic Gaussian. 
Assign Grassmann variables $\eta_1, \cdots, \eta_{2n}$ to the first $n$ and $\tilde \eta_1, \cdots, \tilde \eta_{2n}$ to the last $n$ modes, the Fourier representation is 
\leqalign{eq:maxEntangledState}{
    \Xi_{\rho_I}(\eta, \tilde \eta) 
    = \prod_{j=1}^{2n} \exp(i\eta_j\eta_{2n+j}) 
    = \exp\left(i \eta^T \tilde \eta\right). 
}
First note that the maximally entangled state expression~\eqref{eq:maxEntangledState} has a Gaussian expression since 
\begin{eqnarray}
    \eta^T \tilde \eta = \df 1 2 \begin{pmatrix}
        \eta \\ \tilde \eta 
    \end{pmatrix}^T \begin{pmatrix}
        0 & 1_{2n\times 2n} \\ -1_{2n\times 2n} & 0 
    \end{pmatrix} \begin{pmatrix}
        \eta \\ \tilde \eta 
    \end{pmatrix}. 
\end{eqnarray}
The Choi state of a channel $\Phi:\Cl_{2n}\to \Cl_{2n}$ is $\Phi$ applied to one half of the maximally entangled state. 
\leqalign{def:choiState}{
    \rho_{\Phi} = (\Phi \ot 1_{\Cl_{2n}\to \Cl_{2n}})\rho_I \in \Cl_{4n}. 
}
Since we restrict our attention to unitary channels, denote the Choi state of the unitary channel $\Phi_U$ by $\rho_U$ instead of $\rho_{\Phi_U}$. We begin with a lemma which constrains the real linear transforms of generators available by unitary conjugation. 

\begin{lemma}
    \label{lem:linearOrthogonalConstraint}
    Given a unitary $U\in \Cl_{2n}$ whose conjugate action is a real linear transform of the generators 
     $
            U\gamma_jU^\dag = \sum_{k=1}^{2n} A_{jk}\gamma_k
    $, then $A$ is an orthogonal matrix satisfying $A^TA=I$. 

    \begin{proof}
        Unitary conjugation effects an automorphism of the algebra, so it must preserve the 
        Clifford relation~\eqref{eq:cliffordRelationApp}:  
        \malign{
            2\delta_{jk} 
            &= \{\sum_a A_{ja} \gamma_a, \sum_b A_{kb} \gamma_b\} 
            = \sum_{ab} A_{ja}A_{kb} \{\gamma_a, \gamma_b\}  \\ 
            &= \sum_{ab} A_{ja}A_{kb} 2\delta_{ab} 
            = 2 \sum_a A_{ja} A_{ka},  
        }
        where  $ \delta_{jk} = (A^TA)_{jk}$ implies that $A\in O(2n)$.
    \end{proof}
\end{lemma}

Lemma~\eqref{lem:linearOrthogonalConstraint} does not establish that $U$ is Gaussian: Gaussian unitaries biject with $SO(2n,\real)$, and there exists non-Gaussian unitaries which correspond to an orthongonal matrix $O$ with $\det O=-1$. 
A simple example is $U=\gamma_1$ whose corresponding $2n\times 2n$ orthogonal transformation 
$O$ is $\mrm{diag}([1,-1,-1,...,-1])$ with $\det O=-1$.

\begin{lemma}[Characterization of Gaussian Choi states]
    \label{lem:gaussianChoiCharacterization}
    For any unitary channel 
    $ \Phi_U: \Cl_{2n}\to \Cl_{2n}, \quad \Phi(A) = UAU^\dag$, i.e., has the following form
    the Fourier transform of the Choi state is Gaussian
    \leqalign{eq:unitaryChannelFourier}{
        \Xi_{\rho_U}(\eta, \tilde \eta) = \exp \left[\df i 2 \begin{pmatrix}
            \eta \\ \tilde \eta 
        \end{pmatrix}^T \begin{pmatrix}
            0 & O^T \\ -O & 0 
        \end{pmatrix} \begin{pmatrix}
            \eta \\ \tilde \eta 
        \end{pmatrix}
        \right], \quad OO^T = I,
    }
    if and only if $U$ effects an orthogonal transform $\eta \mapsto O\eta$. 

    \begin{proof}
        Given a $n$-qubit Gaussian $U$ with rotation $R$, the channel $\Phi_U\ot 1_{\Cl_{2n}\to \Cl_{2n}}$ 
        applies $R$ to $\eta$ and leaves $\tilde \eta$ unchanged. Substituting this yields the expression for a Gaussian state~\eqref{eq:unitaryChannelFourier}: 
        \[ 
            \Xi_{\rho_U}(\eta, \tilde \eta) 
            = \Xi_{\rho_I}(R\eta, \tilde \eta) 
            = \exp \left(i \eta^T R^T\tilde \eta\right).
        \]
        Conversely, given a Gaussian state $\Phi_U$ on $2n$ qubits with generators $\eta, \tilde \eta$, 
        \leqalign{eq:gaussian2qubit}{
            \Xi_{\Phi_U}(\eta, \tilde \eta) 
            = \exp \left[\df i 2 \begin{pmatrix}
                    \eta \\ \tilde \eta 
                \end{pmatrix}^T M \begin{pmatrix}
                    \eta \\ \tilde \eta 
                \end{pmatrix}\right], \quad 
            M = \begin{pmatrix}
                A & B^T \\ -B & D 
            \end{pmatrix}, \quad 
            A, D \text{ antisymmetric}, \quad B_{jk}\in \R. 
        }
        Since $\Phi_U$ is a Choi state derived from $\rho_I$ according to equation~\eqref{eq:maxEntangledState}, 
        the Fourier expression satisfies 
        \[
            \Xi_{\Phi_U}(\eta, \tilde \eta) = \Xi_{\rho_I}(\varphi(\eta), \tilde \eta),
        \] 
        for some transformation $\varphi$ corresponding to the unitary action of $U$. The Gaussian expression~\eqref{eq:gaussian2qubit} further constrains $\varphi(\eta)=O\eta$ for some $2n\times 2n$ real matrix $O$, and invoking Lemma~\ref{lem:linearOrthogonalConstraint} shows that $O$ must be an orthogonal matrix. Substituting this simplifies expression~\eqref{eq:gaussian2qubit} to the Gaussian expression~\eqref{eq:unitaryChannelFourier}. 
    \end{proof}
\end{lemma}
Gaussian unitaries biject with rotations $SO(2n, \R)\subsetneq O(2n)$, so to perform a Gaussian unitary channel test we only need to distinguish between $\det O = \pm 1$. To do so, we need the following lemma. 
\begin{lemma}[Special orthogonality criterion]
    \label{lem:specialOrthogonalityCriterion}
    Given $U\in \Cl_{2n}$ which effects an orthogonal transform $\gamma\mapsto O\gamma$, 
    then $U$ is Gaussian iff it is even. 
    \begin{proof}
       Consider the pure diagonal Gaussian 
       \malign{
            \rho_0 
            = (|0\ra \la 0|)^{\ot n} = \df 1 {2^n} \bigotimes_{j=1}^n (1 + Z_j) 
            = \df 1 {2^n} \bigotimes_{j=1}^n (1 - i\gamma_{2j-1}\gamma_{2j}),
       }
       and the corresponding momoent generating operator
        \malign{
            \Xi_{\rho_0}(\eta)
            &= \bigotimes_{j=1}^n (1 - i\eta_{2j-1}\eta_{2j})  
            = \exp \left(-i \sum_{j=1}^n \eta_{2j-1}\eta_{2j}\right) \\ 
            &= \exp\left(\df i 2 \eta^T h \eta\right), \quad \text{where}\quad
            h = -\bigoplus_{j=1}^n \begin{pmatrix}
                0 & -1 \\ 1 & 0 
            \end{pmatrix}. 
       }
       Applying Wick's theorem (Lemma~\ref{cor:wick}) and the Pfaffian property $\Pf(A^TBA) = \det(A)\Pf(B)$, we obtain 
       \malign{
            \Tr{\gamma_{[2n]}^\dag \rho_0} = (\rho_0)_{[2n]} = \mrm{Pf}(h),
       }
       and
       \malign{
            \Tr{\gamma_{[2n]}^\dag U\rho_0 U^\dag} 
            =(U\rho_0 U^\dag)_{[2n]} = \Pf(O^ThO) = \det(O)\Pf(h) 
            = \det(O) \Tr{\gamma_{[2n]}^\dag \rho_0}.
       }
        Note that $\gamma_{[2n]}^\dag \propto Z^{\ot n}$ and $\rho_0$ is nontrivial in the $\gamma_{[2n]}$ component, then 
        \malign{
            \det O = 1 
            &\iff \tr \left[(U^\dag \gamma_{[2n]}^\dag U) \rho_0\right] = \tr(\gamma_{[2n]}^\dag \rho_0) \\ 
            &\iff [U, \gamma^\dag_{[2n]}] = 0 \iff [U, Z^{\ot n}] = 0.
        }
        This establishes that given $O$ orthogonal, $U$ is Gaussian iff $ O$ is a rotation iff $ \det O=1$ iff $ U$ is even. 
    \end{proof}
\end{lemma}
\begin{theorem}[Gaussian unitary channel test]
    A unitary channel $\mca E_U$ is Gaussian if and only if 
    \begin{enumerate}
        \item the Choi state $\rho_U$ passes the fermionic Gaussian state test, and 
        \item $U$ passes the even unitary test in Lemma~\ref{lem:evenUnitaryTest}. 
    \end{enumerate}
    \begin{proof}
        If $U$ is Gaussian, then $\rho_U$ is Gaussian by Lemma~\ref{lem:gaussianChoiCharacterization} and even by definition. Conversely, by Lemma~\ref{lem:gaussianChoiCharacterization} the Choi state $\rho_U$ is Gaussian iff $U$ effects an orthogonal transform of the generators, then Lemma~\ref{lem:specialOrthogonalityCriterion} establishes that $U$ is Gaussian iff it additionally passes the even unitary test. 
    \end{proof}
\end{theorem}
\end{document}